\documentclass[12pt]{article}
\usepackage{amssymb,amsmath,amsfonts,eurosym,geometry,graphicx,caption,color,setspace,sectsty,comment,footmisc,caption,natbib,pdflscape,subfigure,array,hyperref}
\usepackage{lmodern}
\usepackage{booktabs}
\usepackage[T1]{fontenc}
\usepackage{hyperref}
\usepackage{subcaption}

\usepackage{setspace}
\AtBeginDocument{
  \setlength{\abovedisplayskip}{5pt}
  \setlength{\belowdisplayskip}{5pt}
  \setlength{\abovedisplayshortskip}{3pt}
  \setlength{\belowdisplayshortskip}{3pt}
}

\usepackage{makeidx}
\usepackage[dvipsnames]{xcolor}
\usepackage{xcolor}
\usepackage{amsfonts}
\usepackage{amssymb}
\usepackage{amsmath}
\usepackage{tabularx}
\usepackage{float}      
\usepackage{subcaption} 
\usepackage{graphicx}   
\usepackage{multirow}
\usepackage{amstext}
\usepackage{caption}
\usepackage[singlelinecheck=false,labelsep=space,justification=centering]{caption}
\usepackage[flushleft]{threeparttable}
\usepackage{url}

\usepackage{natbib}
\definecolor{mycolor}{RGB}{0,0,128}
\hypersetup{
    colorlinks=true,
    allcolors=mycolor,  
}
\usepackage[title]{appendix}
\usepackage{comment}

\renewcommand{\arraystretch}{.5}

\advance \topmargin by -\headheight
\advance \topmargin by -\headsep
\evensidemargin \oddsidemargin
\let\tilde=\widetilde

\newtheorem{assumption}{Assumption}
\newtheorem{proposition}{Proposition}

\newtheorem{corollary}{Corollary}

\renewcommand{\arraystretch}{1.5}

\providecommand\abstractname{Abstract}
\def\abstract{}

\renewenvironment{abstract}{  \centering\small
  \textbf\abstractname
  \list{}{\leftmargin0.07in \rightmargin\leftmargin}
  \item\relax
}{  \endlist \par\bigskip
}
\makeatletter
\def\@biblabel#1{\hspace*{-\labelsep}}
\makeatother

\newtheorem{remark}{Remark}

\AtBeginDocument{
  \setlength{\topsep}{3pt}
  \setlength{\partopsep}{0pt}
  \setlength{\parskip}{0pt}
}

\begin{document}

\title{Regime-Switching Models for Disaggregated Data\thanks{{This paper is based on the first chapter of Anlong Qin's PhD dissertation at Boston University (Qin, 2022). We thank Pierre Perron, Hiro
Kaido, Ivan Fernandez-Val, Jean-Jacques Forneron for their comments and
suggestions.}}}
\author{
\begin{tabular}{c}
Anlong Qin\footnote{Institute of Finance \& AI, Ningbo University of Finance and Economics, No. 899, Xueyuan Road, Zhejiang, China; qinanlong@nbufe.edu.cn.} \\
[-1.5em]
Ningbo University of Finance and Economics
\end{tabular}
\\[1em]
\begin{tabular}{c}
Zhongjun Qu\footnote{Department of Economics, Boston University, 270 Bay State Road, Boston, MA 02215 USA; qu@bu.edu.} \\
[-1.5em]
Boston University
\end{tabular}
}
\date{\today }
\maketitle

\begin{abstract}
\baselineskip=21.0ptWe show analytically and via simulation that cross-sectional aggregation can substantially attenuate regime-switching signals in time-series data, making regime switches harder to detect. Building on this, we develop regime-switching models and an estimation algorithm which allow for autoregressive dynamics and grouped heterogeneity. We apply the approach to a U.S. macroeconomic dataset of 94 series, covering components of real gross domestic product, industrial production, capacity utilization, employment, and hours worked. The estimates give sharper business cycle classifications than those typically found in the literature. Monte Carlo simulations show that the computation is practical for datasets with a few hundred time series. 

\vspace{5mm}

\vspace{3mm}
\end{abstract}


\thispagestyle{empty}\setcounter{page}{0}

\section{Introduction}

Regime-switching models have become a common framework for capturing nonlinear dynamics in macroeconomic and financial data since \cite{hamilton1989}.\footnote{See \cite{Ang2012} and \cite{hamilton2016} for surveys.} Most applications fit such models to just one series. For example, \cite{hamilton1989} and  \cite{hamilton2006} studied U.S. aggregate output; \cite{masson2001} examined exchange rates; \cite{hamilton1988}, \cite{hamilton2005}, and \cite{Ang2002b} modeled interest rates; \cite{davig2004} analyzed the debt-output ratio; and \cite{Ang2002a} considered equity returns. A small number of studies looked at small-scale multivariate settings: see \cite{diebold1996}, \cite{kim1998}, \cite{chauvet1998}, and \cite{hamilton2006}, who modeled sales, income, employment, and industrial production using a latent factor with regime switching; \cite{Chauvet2002}, who applied a similar model to U.S. unemployment rates across seven age groups; and \cite{krolzig2000}, who estimated a Markov-switching vector error correction model, for six industrial production indices for U.K. business cycles. 

Even for datasets with many series, most analyses treat series separately (e.g., \cite{OwyangPigerWall2005}) rather than formulate a joint model. Only a few papers have studied high-dimensional models: \cite{LiuChen2016} propose a factor model with regime-dependent means, loadings, and covariance structure, estimated via eigenanalysis and the Viterbi algorithm. \cite{UrgaWang2024} develop a quasi-maximum likelihood estimator for factor models with regime-switching loadings, combining an EM algorithm with principal component analysis. \cite{BarigozziMassacci2025} introduce an approximate factor model with regime-switching loadings driven by a latent Markov process, estimated by principal components followed by an EM step.  In a different approach, \cite{HamiltonOwyang2012} develop a Bayesian latent-class model for state-level recessions, modeling recession probabilities through logistic functions tied to national and regional clusters; their results show cross-state differences in recession timing but little evidence that states avoid national downturns. These contributions broaden the scope of regime-switching analysis and suggest potential benefits can come from modeling many series simultaneously.

This study considers regime switching in high-dimensional settings, with a different motivation from prior work. Our starting point is the observation that many macroeconomic variables are built from disaggregated components. For example, GDP combines personal consumption expenditures, business investment, government spending, and imports and  exports, each containing subcomponents.\footnote{Source: U.S. Bureau of Economic Analysis, via
\href{https://fred.stlouisfed.org/release/tables?rid=53\&eid=13498\#snid=13511}
{Federal Reserve Economic Data (FRED)}, Federal Reserve Bank of St.\ Louis.
Percent Change From Preceding Period in Real Gross Domestic Product, Expanded Detail. } The same applies to industrial production, capacity utilization, employment, and hours worked, among many others. In practice, one can work with either aggregate or disaggregate data when modeling regime switching. This raises the question of whether disaggregate data can better identify regime switching than the corresponding aggregates. We study this question analytically and through simulations; both suggest that disaggregate data can lead to sharper identification, with gains increasing in the heterogeneity of signal-to-noise ratios across series. Given the prevalence of aggregation in economic data, this suggests that models and estimation methods designed for disaggregated series can play a useful role and should not be overlooked in practice.

To illustrate this point empirically, we apply a Hamilton (1989) style model to U.S. quarterly real GDP growth from 1972:II to 2019:III. The sample is restricted to the pre-COVID period, so that the pandemic does not dominate the comparison. Let $Y_t$ denote GDP growth. We estimate
\begin{equation}
\label{eq:MS_agg}
Y_t = \delta_1 \mathbf{1}_{\{S_t=1\}} + \mu_2 + e_t,
\end{equation}
where $S_t$ follows a two-state Markov chain, with $S_t=1$ and $S_t=2$ being recession and expansion, respectively; $\mu_2$ is the expansion mean;  $\delta_1$ is the difference between the recession and expansion means; and $e_t \sim \text{i.i.d. } N(0,\sigma_1^2)$ if $S_t=1$ and $e_t \sim \text{i.i.d. } N(0,\sigma_2^2)$ if $S_t=2$. Figure~\ref{fig:Figure1}(a) shows the smoothed recession probabilities, with NBER-defined recessions shown as shaded bars. The recession probabilities are broadly consistent with the NBER chronology; however, the turning points are often not sharply identified, and the model essentially misses the recession in the early 2000s. These patterns also appear in the recession probabilities published by the St.\ Louis Fed, reproduced in Appendix Figure \ref{fig:Fred_recession} for reference.\footnote{Chauvet, Marcelle, and Jeremy M.\ Piger, Smoothed U.S. Recession Probabilities \href{https://fred.stlouisfed.org/series/RECPROUSM156N}{[RECPROUSM156N]}, retrieved from Federal Reserve Economic Data (FRED), Federal Reserve Bank of St.\ Louis, December~31,~2025.} Next, we fit the same model to the components of GDP to see whether there are gains from using disaggregate data. We consider a system of 30 equations, each for a component of GDP (excluding government spending because it is not procyclical):
\begin{equation}
\label{eq:MS}
Y_{i,t}
=
\delta_{i,1}\,\mathbf{1}_{\{S_t=1\}}
+
\mu_{i,2}
+
e_{i,t},
\end{equation}
where $S_t=1$ and $S_t=2$ are as before; $e_{i,t} \sim \text{i.i.d. } N(0,\sigma_{i,1}^2)$ if $S_t=1$ and $e_{i,t} \sim \text{i.i.d. } N(0,\sigma_{i,2}^2)$ if $S_t=2$. The errors $e_{i,t}$ are specified as cross-sectionally independent. Figure~\ref{fig:Figure1}(b) reports the smoothed recession probabilities. All NBER recessions, including the one in the early 2000s, are identified, and the turning points are sharper than in the aggregate case. Since the two models share the same functional form and are estimated using the same sampling algorithm and prior, the difference in results is likely due to disaggregation. In Section~\ref{sec:application} we show that still sharper results can be obtained by incorporating additional disaggregated time series in the estimation.

\begin{figure}[H]
    \centering
    
    \caption{Inference on recessions using U.S.\ GDP or its components}
    \label{fig:Figure1}
    \begin{minipage}[t]{0.5\textwidth}
        \centering
        {\small (a) Based on GDP}\\
        \includegraphics[width=\linewidth, clip, trim=1.7cm 9.5cm 1.2cm 10cm]{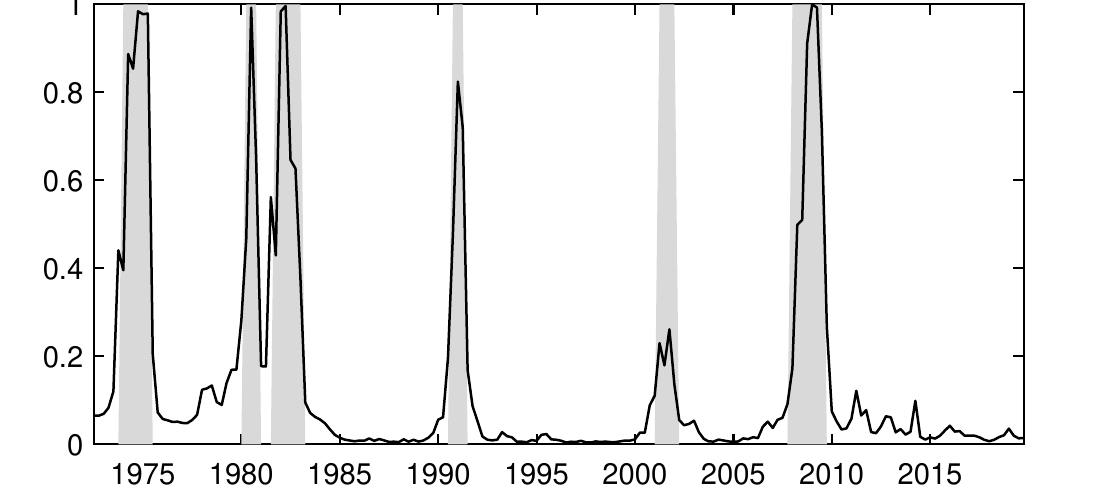}
    \end{minipage}%
    \begin{minipage}[t]{0.5\textwidth}
        \centering
        {\small (b) Based on GDP components}\\
        \includegraphics[width=\linewidth, clip, trim=1.7cm 9.5cm 1.2cm 10cm]{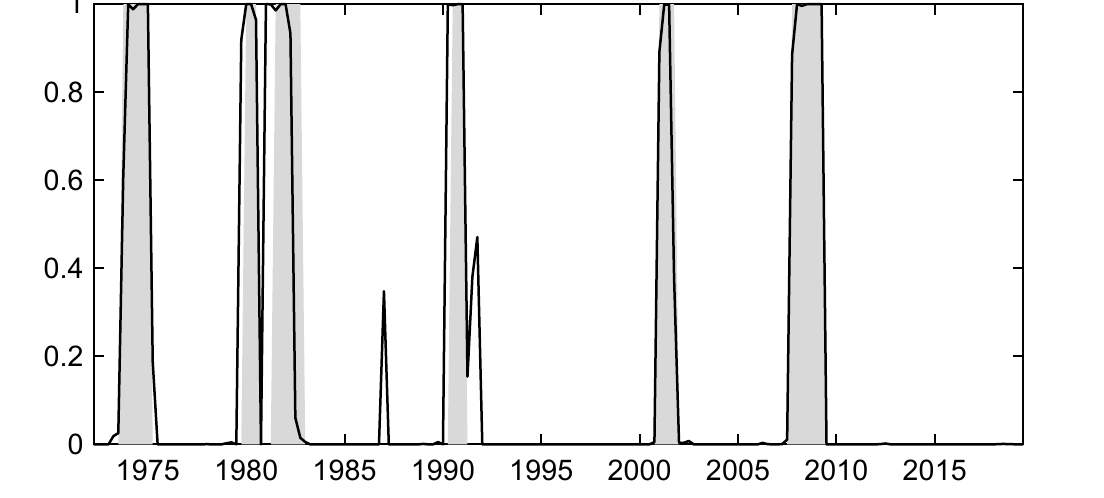}
    \end{minipage}
    \vspace{-0.4cm}
\begin{minipage}{\linewidth}
\singlespacing
\footnotesize{ \textit{Note:} The aggregate GDP series is modeled as in (\ref{eq:MS_agg}), and each of the 30 GDP components as in (\ref{eq:MS}). 
Smoothed probabilities are obtained by averaging 0.4 million MCMC draws. 
NBER recession dates are shown as shaded regions. 
Priors are described in Section~3.5 and are the same across applications.}
\end{minipage} 
    \end{figure}

The above comparison shows that disaggregated data can improve regime detection even without changing the model in an essential way. The reason is that the components of GDP differ in their signal-to-noise ratios with respect to regime switching. Aggregation ignores this heterogeneity and fails to downweight noisier series, while disaggregated data allow the likelihood to adapt accordingly. In Section~\ref{sec:effects of aggregation}, using a simple two-variable example, we show analytically that aggregation preserves all information about regime switching if and only if the signal-to-noise ratios are equal across the two series. When they differ, as in the present application, the information loss can be substantial. We also quantify the gains through simulations.

Our proposed models are generalizations of \eqref{eq:MS} that allow for autoregressive dynamics and grouped heterogeneity. In a typical application, the model contains hundreds of parameters, and standard maximum likelihood often fails to converge. To address this, we develop a Gibbs sampling algorithm for estimation, inference, and forecasting, which incorporates a new multi-step sampler for latent state variables. We show that the algorithm can handle large macroeconomic datasets (several hundred time series) in reasonable computation time.

We apply the models to a dataset of 94 disaggregated series, covering components of real GDP, industrial production, capacity utilization, employment, and hours worked. The models detect all NBER-defined recessions over 1972-2024. The smoothed recession probabilities are close to being binary, taking values near zero or one with few intermediate values. This contrasts with typical estimates in the literature, such as those in Figure \ref{fig:Fred_recession}, which tend to linger at intermediate values for extended periods. The estimated recession starting dates either coincide exactly with the NBER dates, or in a few cases, lag them by one to three quarters. Meanwhile, the estimated recession ending dates closely track the starts of NBER expansions. This asymmetry may reflect the nature of business cycles: downturns build gradually across sectors, while recoveries tend to generate a more synchronized signal across disaggregated series. These findings suggest that methods designed around a wide set of disaggregated variables can be useful for business cycle research, particularly for accurate regime identification.

The issue of cross-sectional aggregation has been studied mainly for linear models such as ARMA models and for forecasting; see, e.g., \cite{lutkepohl1987}, \cite{Granger1988}, \cite{Luetkepohl2009}. To our knowledge, this is the first formal analysis of aggregation in regime-switching models. Here, the object of interest is a common regime component shared across series rather than a linear forecast. Our results suggest that disaggregation gives clear gains in this setting. We conjecture that the same insight extends to other settings where a shared latent component is of interest, such as nonlinear factor models or models with common stochastic volatility. In those cases, working with less aggregated series may produce similar gains.

The paper is organized as follows. Section 2 introduces the models and examines gains from disaggregation. Section 3 presents a Gibbs sampler for estimation, filtering, and forecasting. Section 4 studies the methods through simulations. Section 5 presents an empirical application, and Section 6 concludes. Additional figures and details are available in the online appendix.

\section{Proposed models}

We consider $N$ time series partitioned into $J$ groups. Denote the series in group $j$ by
\[
y_{j,t} = (y_{j,1,t}, y_{j,2,t}, \ldots, y_{j,N_j,t})^{\prime},
\]
where $N_j$ is the number of series in group $j$, and
$N = \sum_{j=1}^{J} N_j$.  If each group contains only a single series, then $J = N$. The full data vector at time $t$ (for $t = 1,\ldots,T$) is
\begin{equation}
y_t = (y_{1,t}^{\prime}, y_{2,t}^{\prime}, \ldots, y_{J,t}^{\prime})^{\prime}.
\label{dis}
\end{equation}
Let \(
Y_{j,t} = \sum_{i=1}^{N_j} y_{j,i,t}
\) be the sum of the
series in group $j$. The vector of aggregated series at $t$ is
\begin{equation}
Y_t = (Y_{1,t}, Y_{2,t}, \ldots, Y_{J,t})^{\prime}.
\label{agg}
\end{equation}
Our focus is on the disaggregate series in
\eqref{dis}; the aggregate series in \eqref{agg} are included for comparison. Throughout, uppercase letters denote aggregate
variables.


We propose two models. Each consists of a common regime-switching component and an autoregressive component that captures additional dynamics. Some coefficients are equal within each group  to capture grouped heterogeneity. Both the mean and variance can switch.

The first model is a generalization of \cite{hamilton1989}  to
the disaggregated setting:
\begin{equation}
\text{\bf{\textsc{Model A}}: \qquad }\;
y_{j,i,t}
=
\mu_{j,i,S_t}
+
\sum_{m=1}^{k}
\phi_{j,m}
(y_{j,i,t-m}-\mu_{j,i,S_{t-m}})
+
e_{j,i,t},
\label{eq:2.1}
\end{equation}
where $(j,i,t)$ indexes the $i$-th unit in the $j$-th group at
time $t$; $k$ is the lag order; and $(S_{t-k},\ldots,S_t)$ is a vector of
 latent variables determining the state of the system at $t$.
In general, $S_t$ may take values in a finite state space
$S_t\in\mathcal{S}=\{1,2,\ldots,S\}$, with known $S$; in the leading case,
$S_t$ is binary, representing recession and expansion. We assume $S_t$
follows a Markov process with transition probabilities $p_{ss'}$ for any
$s,s'\in\mathcal{S}$. The errors are assumed to be normally distributed:
\[
e_{j,i,t}\sim N(0,\sigma_{j,i,S_t}^2).
\]
The unknown parameters consists of the
group-level autoregressive coefficients $\phi_{j,m}$, the regime-dependent
means $\mu_{j,i,S_{t}}$, and the regime-dependent variances
$\sigma_{j,i,S_t}^2$. 

A known feature of regime-switching models is label switching:
without an ordering restriction, the regimes are not identified. We impose the following ordering restriction on the means:
\[
\mu_{j,i,1} \le \mu_{j,i,2} \le \cdots \le \mu_{j,i,S},
\qquad \text{for any }  j \text{ and } i .
\]
Because $S_t$ is a scalar process, restricting the
mean of a single series is in fact sufficient for identification; e.g., one could require the first series in the
first group to satisfy
$\mu_{1,1,1} \le \mu_{1,1,2} \le \cdots \le \mu_{1,1,S}$. In our applications,
however, we find that imposing the restriction for all series improves computational performance. For variables such as
the unemployment rate, whose values are higher during recessions, we multiply
by $-1$ so that they satisfy these inequalities.

In the two-regime case, it is convenient to write
\[
\mu_{j,i,S_t}
=
\delta_{j,i,1}\mathbf{1}_{\{S_t=1\}}+\mu_{j,i,2},
\qquad
\delta_{j,i,1}\le 0 \ \text{for all } i,j .
\]
Under this parameterization, $\mu_{j,i,2}$ is unrestricted, and each ordering restriction involves a single coefficient rather than two. The reduction is particularly useful in high-dimensional settings, as it removes $N$ parameters from the set of coefficients subject to restrictions during estimation.

The aggregate model implied by Model A is
\begin{equation}
Y_{j,t}
=
\mu_{j,S_t}
+
\sum_{m=1}^{k}
\phi_{j,m}(Y_{j,t-m}-\mu_{j,S_{t-m}})
+
e_{j,t},
\label{eq:2.2}
\end{equation}
where $e_{j,t}\sim N(0,\sigma_{j,S_t}^2$),
\(
\mu_{j,S_t}
=\sum_{i=1}^{N_j}\mu_{j,i,S_t}$, and
$\sigma_{j,S_t}^2
=\sum_{i=1}^{N_j}\sigma_{j,i,S_t}^2.
\)

\begin{remark}
The model allows the variance to switch so that it can accommodate both mild and severe recessions, such as those of 2001 and 2008. Without this flexibility, it can become difficult to separate mild recessions from expansions when severe recessions appear in the sample. To see why, suppose the variance is restricted to be constant across regimes. The likelihood of the recession regime is then relatively peaked (i.e., has a low variance) because the variance is determined largely by expansion observations, which constitute the vast majority of the sample. As a result, a single large negative observation can pull the recession mean toward an extreme value, leaving mild recessions closer to the expansion mean and harder to classify. Allowing the variance to switch lets the recession regime variance take on a larger value, flattening its likelihood and reducing the influence of any single observation on the recession mean. This preserves the separation between mild recessions and expansions.
\end{remark}

\begin{remark}
The AR coefficients are assumed to be equal within each group, which reduces the number of parameters relative to letting them vary freely. Beyond parsimony, this assumption is motivated by the fact that a Markov chain has an AR(1) representation with heteroskedastic errors \citep{hamilton1994}. Therefore, giving the AR component too much flexibility may lead to weak regime identification. The restriction also implies that aggregated series can be obtained by summing disaggregated ones within each group, which facilitates the theoretical analysis.
\end{remark}

The second model generalizes the univariate models in \cite{hansen1992}, \cite{Cho2007}, and \cite{qu2021}:
\begin{equation}
\text{\bf{\textsc{Model B}}: \qquad}\;
y_{j,i,t}
=
\mu_{j,i,S_t}
+
\sum_{m=1}^{k}
\phi_{j,m} y_{j,i,t-m}
+
e_{j,i,t}.
\label{eq:2.3}
\end{equation}
The error terms $e_{j,i,t}$, the latent state $S_{t}$, and the parameters ($\phi_{j,m}$, $\mu_{j,i,S_{t}}$, $\sigma_{j,i,S_t}^2$) are
the same as in Model A. The corresponding aggregate model is
\begin{equation}
Y_{j,t}
=
\mu_{j,S_t}
+
\sum_{m=1}^{k}
\phi_{j,m} Y_{j,t-m}
+
e_{j,t},
\label{eq:2.4}
\end{equation}
where $\mu_{j,S_t}$, $\phi_{j,m}$, and $e_{j,t}$ are defined as before. 

Models A and B share similar second-order dynamics but are not observationally equivalent. For example, when $k=1$, both models have an ARMA(1,1) representation via the AR(1) representation of the latent Markov chain. They differ in how the series reacts to regime changes. Model A implies an immediate jump toward the new regime, while Model B implies a more gradual transition as the lagged level $y_{t-1}$ carries information from the previous regime. As a result, they imply different higher-order moments, as in the univariate setting.

Models A and B are designed for situations where average growth rates determine regimes. In other settings, regimes may depend on other parameters, such as a set of regression coefficients. Model~B can be generalized to such applications. For example, one may consider
\[
y_{j,i,t}
=
x_{j,i,t}'\beta_{j,i}
+
z_{j,i,t}'\delta_{j,i,S_t}
+
e_{j,i,t},
\]
where $x_{j,i,t}$ and $z_{j,i,t}$ are random vectors. The coefficients $\beta_{j,i}$ remain constant and may exhibit a grouped structure, while $\delta_{j,i,S_t}$ are regime dependent. The error term $e_{j,i,t}$ and the latent states are defined as in Model B. The proposed algorithm can be adapted to estimate this model.

\subsection{Assumptions}

The following assumptions are familiar in the Markov-switching literature and ensure that the models are identified and that estimation is feasible. Their main role is to help assess whether the models are suitable for the intended applications.

\begin{assumption}
The Markov process $S_t$ affects all disaggregate series $y_{j,i,t}$ simultaneously.
\end{assumption}

This assumption implies a common
cycle affects all data units with the same timing. It is nonrestrictive
when all series are coincident indicators of the business cycle, such as
changes in non-farm employment, industrial production, capacity utilization,
real personal income, and real manufacturing and trade sales. It
is violated when coincident indicators are combined with leading
indicators, such as stock market indices or the term structure of interest
rates. It implies that the disaggregate variables should be chosen so
that they respond to $S_t$ with approximately the same timing. Economic
theory and empirical evidence can often guide this choice.

\begin{assumption}
The process $S_t$ is ergodic, with
$S_t \in \mathcal{S}=\{1,\ldots,S\}$ and $S>1$.
\end{assumption}

This assumption implies that every state can be reached from every other state with positive probability. This property ensures the existence of a unique stationary distribution and is necessary for identification and inference based on the likelihood function.

\begin{assumption}
The regime means satisfy
\(
\mu_{j,i,1} \le \mu_{j,i,2} \le \cdots \le \mu_{j,i,S},
\)
for $i\in\{1,\ldots,N_j\}$ and $j\in\{1,\ldots,J\}$, with at least one of the
inequalities being strict for at least one time series.
\end{assumption}

This assumption addresses the label-switching property of regime-switching models. Requiring at least one strict inequality ensures that the regimes are distinguishable in the data without restricting heterogeneity across series and groups.

\begin{assumption}
The error terms $\{e_t\}$ are conditionally normally distributed with
$E(e_t\mid S_t)=0$ and
\(
E(e_t e_t' \mid S_t)=V_{S_t}
=\operatorname{diag}\{\sigma_{1,S_t}^2,\sigma_{2,S_t}^2,\ldots,
\sigma_{N,S_t}^2\}.
\)
\end{assumption}

The normality assumption reflects the parametric nature of the regime-switching model and is standard in the literature. Setting the off-diagonal elements of the covariance matrix to zero implies that the joint likelihood factorizes as the product of the marginal likelihoods of the individual series. This is analogous to a composite likelihood approach: cross-sectional dependence in the errors is intentionally not modeled, which may sacrifice estimation efficiency but reduces computational costs and the risk of misspecification.\footnote{See \citet{varin2011overview} for a review of composite likelihoods.}

\begin{assumption}
(i) There exists $0<C<\infty$ with $|\mu_{j,i,s}|\le C$ and $C^{-1}\le \sigma_{j,i,s}^2 \le C$ for $i\in\{1,\ldots,N_j\}$,
$j\in\{1,\ldots,J\}$, and $s\in\mathcal{S}$;
(ii) The roots of the lag polynomials
\(
\Phi_j(L)=1-\phi_{j,1}L-\phi_{j,2}L^2-\cdots-\phi_{j,k}L^k
\)
lie outside the unit circle for $j\in\{1,\ldots,J\}$;
(iii) $p_{ss'}\in[0,1)$ for $s,s'\in\mathcal{S}$.
\end{assumption}

This assumption places mild restrictions on the parameters. It rules out degenerate cases by requiring finite means and positive variances.  The stability condition on the lag polynomials implies stationary autoregressive dynamics within each regime. These conditions are needed for the numerical stability of the proposed estimation algorithm.

\subsection{Effects of cross-sectional aggregation on state recovery}
\label{sec:effects of aggregation}
We consider a simple setting, with two disaggregated time series and known parameters, to examine the effects of aggregation on the identification of recessions. This can be viewed as an asymptotic exercise in which all parameters are assumed to have been consistently estimated. In finite samples, the disaggregated model involves more parameters and greater estimation uncertainty. Simulations later in the paper examine this issue.

Regime switching in mean and in variance each carries useful information. We consider the two cases separately. Lagged dependent variables (as in Model B) do not affect regime updating, since the autoregressive coefficients do not switch; they are omitted.

\paragraph{Case 1: Switching in the mean only.}

We consider two time series given by
\begin{equation}
\begin{aligned}
y_{1,t}
&=
\mu_{1,1}\mathbf 1_{\{S_t=1\}}
+
\mu_{1,2}\mathbf 1_{\{S_t=2\}}
+
e_{1,t},
\qquad
e_{1,t}\sim N(0,\sigma_1^2),
\\
y_{2,t}
&=
\mu_{2,1}\mathbf 1_{\{S_t=1\}}
+
\mu_{2,2}\mathbf 1_{\{S_t=2\}}
+
e_{2,t},
\qquad
e_{2,t}\sim N(0,\sigma_2^2).
\end{aligned}
\label{eq:mean_switch_model}
\end{equation}
The aggregated series is
\(
Y_{t}=\mu_{1}\mathbf{1}_{\{S_{t}=1\}}+\mu_{2}\mathbf{1}_{\{S_{t}=2\}}+e
_{t}, \ \ e _{t}\sim N(0,\sigma ^{2}),
\)
where $Y_{t}=y_{1,t}+y_{2,t}$, $\mu_{1}=\mu _{1,1}+\mu _{2,1}$, $\mu_{2}=\mu
_{1,2}+\mu _{2,2}$, and $\sigma ^{2}=\sigma _{1}^{2}+\sigma _{2}^{2}$. Assume $\mu _{1,1}\leq\mu _{1,2}$ and $\mu _{2,1}\leq\mu _{2,2}$. 

For the disaggregated model, let $\hat{S}%
_{t|t-1}^{1}\equiv E(\mathbf{1}_{\{S_{t}=1%
\}}|y_{t-1},y_{t-2},\ldots)$ be the predicted probability of recession, with conditional mean squared error $\hat{P}%
_{t|t-1}^{1}\equiv E\{(\mathbf{1}_{\{S_{t}=1\}}-\hat{S}%
_{t|t-1}^{1})^{2}|y_{t-1},y_{t-2},\ldots\}=\hat{S}_{t|t-1}^{1}(1-\hat{S}%
_{t|t-1}^{1})$. Similarly, for the aggregated model, define $\Breve{S}_{t|t-1}^{1}\equiv E(\mathbf{1}_{\{S_{t}=1%
\}}|Y_{t-1},Y_{t-2},\ldots)$ and $\Breve{P}%
_{t|t-1}^{1}\equiv E\{(\mathbf{1}_{\{S_{t}=1\}}-\Breve{S}%
_{t|t-1}^{1})^{2}|Y_{t-1},Y_{t-2},\ldots\}$. We use $\hat{S}^1_{t\mid t}$ and $\Breve{S}^1_{t\mid t}$ for filtered probabilities, and
$\hat{S}^1_{t\mid T}$ and $\Breve{S}^1_{t\mid T}$ for smoothed probabilities. The following result characterizes when the aggregated and disaggregated models are equivalent, and when they are not.

\begin{proposition} [Equivalence of filters with switching in the mean only] Suppose the data are generated by (\ref{eq:mean_switch_model}) and the initial conditions of the disaggregated and aggregated models are the same: $\hat{S}%
_{1|0}^{1}=\Breve{S}_{1|0}^{1}=p_{1}$ with $0<p_{1}<1$. Then, $\hat{S}_{t|t}^{1}=\Breve{S}%
_{t|t}^{1}$, $\hat{S}_{t|t-1}^{1}=\Breve{S}_{t|t-1}^{1}$, and $\hat{S}%
_{t|T}^{1}=\Breve{S}_{t|T}^{1}$ for all $t\geq 1$ if and only if 
\begin{equation}
\frac{\mu _{1,1}-\mu _{1,2}}{\sigma _{1}^{2}}=\frac{\mu _{2,1}-\mu _{2,2}}{%
\sigma _{2}^{2}}.  \label{cond-1}
\end{equation}
\end{proposition}

The proposition implies that unless the signal-to-noise ratios are equal as in \eqref{cond-1}, the two filters differ with probability one. This is because filtering is driven by the log-likelihood ratio comparing the two regimes. In the disaggregated model, the ratio is
\[
\ell_t^{(d)}
= \text{const}
+ \frac{\delta_1}{\sigma_1^2}\,(y_{1,t}-\mu_{1,1})
+ \frac{\delta_2}{\sigma_2^2}\,(y_{2,t}-\mu_{2,1}),
\qquad
\delta_i=\mu_{i,1}-\mu_{i,2},
\]
where $y_{1,t}$ and $y_{2,t}$ provide regime information through two separate weights $\delta_1/\sigma_1^2$ and $\delta_2/\sigma_2^2$. In the aggregate model, the ratio is
\[
\ell_t^{(a)}
= \text{const}
+ \frac{\delta_1+\delta_2}{\sigma_1^2+\sigma_2^2}
\bigl[(y_{1,t}+y_{2,t})-(\mu_{1,1}+\mu_{2,1})\bigr],
\]
and regime evidence enters only through the sum, i.e., with equal weights on $y_{1,t}$ and $y_{2,t}$. If 
condition \eqref{cond-1} holds, the two disaggregated weights are equal, and the disaggregated likelihood ratio reduces to a function of $Y_t = y_{1,t}+y_{2,t}$ alone. Otherwise, there exist movements in $(y_{1,t},y_{2,t})$ that change the disaggregated likelihood ratio but leave the aggregate unchanged, and aggregation discards information useful for filtering. The Appendix contains the formal proof.

The next result shows that the disaggregated filter is weakly more precise than the aggregate filter, and strictly more precise when \eqref{cond-1} does not hold. The precision is measured in terms of the mean squared error: for the disaggregated model $E\{(\mathbf{1}_{\{S_{t}=1\}}-\hat{S}%
_{t|t}^{1})^{2}\}=E\{\hat{S}_{t|t}^{1}(1-\hat{S}%
_{t|t}^{1})\}$, and for the aggregrated model $E\{(\mathbf{1}_{\{S_{t}=1\}}-\Breve{S}%
_{t|t}^{1})^{2}\}=E\{\Breve{S}_{t|t}^{1}(1-\Breve{S}%
_{t|t}^{1})\}$.
\begin{corollary}[Filter precision with switching in the mean only] Suppose the data are generated by (\ref{eq:mean_switch_model}) and the initial conditions for the disaggregated and aggregated models are the same: $\hat{S}%
_{1|0}^{1}=\Breve{S}_{1|0}^{1}=p_{1}$ with $0<p_{1}<1$.
Then, for any $t$,
\begin{equation}
E\!\left[\Breve S^1_{t\mid t}\bigl(1-\Breve S^1_{t\mid t}\bigr)\right]
\;\ge\;
E\!\left[\hat S^1_{t\mid t}\bigl(1-\hat S^1_{t\mid t}\bigr)\right].
\label{eq:aggregation_loss}
\end{equation}
Moreover, if \eqref{cond-1} fails, then the inequality is strict for at least one $t$
(in fact, already at $t=1$).
\end{corollary}

To supplement the theoretical result, we consider a simulation experiment on how relative filtering precision depends on the signal-to-noise ratio in (\ref{cond-1}). We report the mean squared error (MSE) and mean absolute error (MAE) of the filtered probabilities for $S_t=1$. For the disaggregated model, these are given by
$
E\{(\hat{S}_{t\mid t}^1-\mathbf{1}_{\{S_t=1\}})^2\}
$ and
$E\{\lvert \hat{S}_{t\mid t}^1-\mathbf{1}_{\{S_t=1\}}\rvert\},
$
and for the aggregated model by
$
E\!\{(\Breve{S}_{t\mid t}^1-\mathbf{1}_{\{S_t=1\}})^2\}
$ and
$E\{\lvert \Breve{S}_{t\mid t}^1-\mathbf{1}_{\{S_t=1\}}\rvert\}$.
All quantities are averaged over 1{,}000 simulation replications. To cover a range of cases, we set $\sigma_{1}^{2}=1$, 
$\mu_{1,2}-\mu_{1,1}=\mu_{2,2}-\mu_{2,1}=2$, and let $\sigma_{2}^{2}$ vary from $0.1$ to $10$. 
Two extreme cases, $\sigma_{2}^{2}=0.01$ and $\sigma_{2}^{2}=100$, are also considered. 
The transition probabilities are set to $p_{11}=0.75$ and $p_{22}=0.95$, consistent with typical estimates for U.S. business cycles.

The results are reported in Table \ref{table:agg_mean_only}.
When $\sigma_2^2=1$, the two models perform identically, as condition \eqref{cond-1} predicts. Their differences increase monotonically as $\sigma_2^2$ decreases from 1. When $\sigma_2^2=0.5$, the MSE under disaggregation is about 18\% lower, with a similar reduction in MAE. When $\sigma_2^2=0.2$, the improvement increases to 78\% under both loss measures; when $\sigma_2^2=0.1$, it reaches 97\%. When $\sigma_2^2>1$, the disaggregated model also does better, where the MSE reduction increases from 11\% at $\sigma_2^2=2$ to 46\% at $\sigma_2^2=10$, and reaches 60\% when $\sigma_2^2=100$. Even with $\sigma_2^2=10$, the standard deviation of the second series is only $\sqrt{10}$ times that of the first, a difference that remains empirically relevant. In the extreme case $\sigma_2^2=100$, the second series contains essentially no regime information. The disaggregated filter adapts to this by downweighting the noisy series, and its MSE is only twice the $\sigma_2^2=1$ benchmark. The aggregated filter becomes completely uninformative. The results show that aggregation can lead to substantial loss when signal-to-noise ratios differ across series. 

\begin{table}[H]
\centering
\caption{MSE and MAE of inference on $S_t$, switching-mean case.}
\label{table:agg_mean_only}

\renewcommand{\arraystretch}{0.54}
\begin{tabular}{c c c c c c c c c}
\toprule
& \multicolumn{4}{c}{\text{MSE}} 
& \multicolumn{4}{c}{\text{MAE}} \\
\cmidrule(lr){2-5} \cmidrule(lr){6-9}
$\sigma_2^{2}$ 
& Dis.
& Agg.
& Imp.
& \% Imp. 
& Dis.
& Agg.
& Imp.
& \% Imp. \\
\midrule
\multicolumn{9}{c}{\text{Benchmark: $\sigma_2^{2}=1$}} \\
\midrule
1.00 & 0.0255 & 0.0255 & 0.0000 & 0.0 & 0.0514 & 0.0514 & 0.0000 & 0.0 \\
\midrule
\multicolumn{9}{c}{\text{Panel A: $\sigma_2^{2}<1$}} \\
\midrule
0.90 & 0.0235 & 0.0236 & 0.0001 & 0.4 & 0.0474 & 0.0476 & 0.0002 & 0.4 \\
0.80 & 0.0217 & 0.0220 & 0.0003 & 1.4 & 0.0436 & 0.0443 & 0.0007 & 1.6 \\
0.70 & 0.0193 & 0.0202 & 0.0009 & 4.5 & 0.0388 & 0.0407 & 0.0019 & 4.7 \\
0.60 & 0.0167 & 0.0185 & 0.0018 & 9.7 & 0.0336 & 0.0372 & 0.0036 & 9.7 \\
0.50 & 0.0135 & 0.0165 & 0.0030 & 18.2 & 0.0273 & 0.0334 & 0.0061 & 18.3 \\
0.40 & 0.0101 & 0.0150 & 0.0049 & 32.7 & 0.0203 & 0.0301 & 0.0098 & 32.6 \\
0.30 & 0.0061 & 0.0129 & 0.0068 & 52.7 & 0.0124 & 0.0261 & 0.0137 & 52.5 \\
0.20 & 0.0024 & 0.0112 & 0.0087 & 78.0 & 0.0049 & 0.0224 & 0.0175 & 78.1 \\
0.10 & 0.0002 & 0.0091 & 0.0089 & 97.8 & 0.0003 & 0.0185 & 0.0182 & 98.4 \\
0.01 & 0.0000 & 0.0077 & 0.0077 & 100.0 & 0.0000 & 0.0155 & 0.0155 & 100.0 \\
\midrule
\multicolumn{9}{c}{\text{Panel B: $\sigma_2^{2}>1$}} \\
\midrule
2.00   & 0.0357 & 0.0400 & 0.0043 & 10.9 & 0.0718 & 0.0808 & 0.0090 & 11.1 \\
3.00   & 0.0399 & 0.0512 & 0.0113 & 22.1 & 0.0806 & 0.1035 & 0.0229 & 22.1 \\
4.00   & 0.0426 & 0.0599 & 0.0173 & 28.9 & 0.0861 & 0.1214 & 0.0353 & 29.1 \\
5.00   & 0.0439 & 0.0670 & 0.0231 & 34.5 & 0.0889 & 0.1358 & 0.0469 & 34.5 \\
6.00   & 0.0452 & 0.0728 & 0.0276 & 38.1 & 0.0913 & 0.1478 & 0.0565 & 38.2 \\
7.00   & 0.0457 & 0.0776 & 0.0319 & 41.1 & 0.0926 & 0.1576 & 0.0650 & 41.2 \\
8.00   & 0.0463 & 0.0812 & 0.0349 & 43.0 & 0.0939 & 0.1655 & 0.0716 & 43.3 \\
9.00   & 0.0470 & 0.0850 & 0.0380 & 44.7 & 0.0951 & 0.1730 & 0.0779 & 45.1 \\
10.00  & 0.0472 & 0.0877 & 0.0405 & 46.2 & 0.0957 & 0.1790 & 0.0833 & 46.5 \\
100.00 & 0.0507 & 0.1256 & 0.0749 & 59.6 & 0.1027 & 0.2579 & 0.1552 & 60.2 \\
\bottomrule
\end{tabular}
\begin{minipage}{1\linewidth}
\footnotesize
\singlespacing
\emph{Notes:} 
``Dis.''\ and ``Agg.''\ refer to disaggregated and aggregated data, respectively. 
``Imp.''\ denotes absolute improvement of Dis.\ over Agg.; ``\% Imp.''\ denotes percentage improvements and is computed relative to Agg.
\end{minipage}
\end{table}

\paragraph{Case 2: Switching in the variance only.}

We consider two time series given by
\begin{align}
y_{1,t} &= \mu_{1} + e_{1,t}, \ \ \ e_{1,t} \sim N\!\left(0,\,
\sigma_{1,1}^{2}\mathbf{1}_{\{S_t=1\}} + \sigma_{1,2}^{2}\mathbf{1}_{\{S_t=2\}}
\right), \nonumber \\
y_{2,t} &= \mu_{2} + e_{2,t},  \ \ \ e_{2,t} \sim N\!\left(0,\,
\sigma_{2,1}^{2}\mathbf{1}_{\{S_t=1\}} + \sigma_{2,2}^{2}\mathbf{1}_{\{S_t=2\}}
\right).
\label{eq:disaggregated_var_only}
\end{align}
The aggregate model is 
$
Y_{t}=\mu+e _{t}$, with $e _{t}\sim N(0,\sigma _{1}^{2}\mathbf{1}%
_{\{S_{t}=1\}}+\sigma _{2}^{2}\mathbf{1}_{\{S_{t}=2\}}),
$
where $Y_{t}=y_{1,t}+y_{2,t}$, $\mu=\mu _{1}+\mu _{2}$, $\sigma
_{1}^{2}=\sigma _{1,1}^{2}+\sigma _{2,1}^{2}$, and $\sigma _{2}^{2}=\sigma
_{1,2}^{2}+\sigma _{2,2}^{2}$. Assume $\sigma_{1,1}\geq \sigma_{1,2}$ and $\sigma_{2,1}\geq \sigma_{2,2}$.

\begin{proposition} [Equivalence of filters with switching in the variance only] Suppose the data are generated by \eqref{eq:disaggregated_var_only} and the initial conditions for the disaggregated and aggregated models are the same: $\hat{S}%
_{1|0}^{1}=\Breve{S}_{1|0}^{1}=p_{1}$ with $0<p_{1}<1$. Then, $\hat{S}_{t|t}^{1}=%
\Breve{S}_{t|t}^{1}$, $\hat{S}_{t|t-1}^{1}=\Breve{S}_{t|t-1}^{1},$ and $\hat{%
S}_{t|T}^{1}=\Breve{S}_{t|T}^{1}$ for all $t\geq 1$ if and only if 
\begin{equation*}
\sigma _{1,1}^{2}=\sigma _{1,2}^{2}\text{ and }\sigma _{2,1}^{2}=\sigma
_{2,2}^{2}\text{,}
\end{equation*}%
that is, there is no regime switching in the variance of either series.
\end{proposition} 

The proposition implies that the filters are different whenever variances switch. In the disaggregated model, the log-likelihood ratio comparing the two regimes is
\[
\ell_t^{(d)}
=
\text{const}
+
\tilde y_{1,t}^{2}
\left(
\frac{1}{2\sigma_{1,1}^{2}}-\frac{1}{2\sigma_{1,2}^{2}}
\right)
+
\tilde y_{2,t}^{2}
\left(
\frac{1}{2\sigma_{2,1}^{2}}-\frac{1}{2\sigma_{2,2}^{2}}
\right),
\]
where $\tilde y_{i,t}=y_{i,t}-\mu_i$, and each series provides regime
information through its own squared deviation $\tilde y_{i,t}^2$. In the aggregated model, the log-likelihood ratio is
\[
\ell_t^{(a)}
=
\text{const}
+
(\tilde y_{1,t}+\tilde y_{2,t})^{2}
\left(
\frac{1}{2(\sigma_{1,1}^{2}+\sigma_{2,1}^{2})}
-
\frac{1}{2(\sigma_{1,2}^{2}+\sigma_{2,2}^{2})}
\right),
\]
which depends on the data only through the squared deviation of the aggregate:
$(Y_t-\mu)^2=(\tilde y_{1,t}+\tilde y_{2,t})^2$. For the two filters to coincide, the two terms in $\ell_t^{(d)}$ must combine 
into a multiple of $(\tilde y_{1,t}+\tilde y_{2,t})^2$. However, this cannot hold except in the trivial case of no-switching: a weighted sum 
$a\,\tilde y_{1,t}^2 + b\,\tilde y_{2,t}^2$ matches 
$c(\tilde y_{1,t}+\tilde y_{2,t})^2$ for all $(\tilde y_{1,t}, \tilde y_{2,t})$ 
only when $c = 0$ (since the latter contains an additional cross term 
$2c\,\tilde y_{1,t}\tilde y_{2,t}$ which must be set to zero), 
and this forces $a = b=0$. Even when the aggregate variance is constant across regimes, 
i.e., $\ell_t^{(a)}$ carries no regime information, the disaggregated 
likelihood continues to distinguish states through the relative scaling of 
individual squared deviations. Equivalence happens only in the 
degenerate case in which variances do not switch at all. The next result confirms that the disaggregated filter is strictly more accurate whenever at least one variance is switching.

\begin{corollary}[Filter precision with switching in the variance only]
Suppose the data are generated by \eqref{eq:disaggregated_var_only} and the initial
conditions for the disaggregated and aggregated models are the same: $\hat{S}_{1\mid 0}^{1}=\Breve{S}_{1\mid 0}^{1}=p_{1}$
with $0<p_{1}<1$. Then for each $t$:
\begin{equation}
E\!\left[\Breve S^1_{t\mid t}\bigl(1-\Breve S^1_{t\mid t}\bigr)\right]
\;\ge\;
E\!\left[\hat S^1_{t\mid t}\bigl(1-\hat S^1_{t\mid t}\bigr)\right].
\label{eq:aggregation_loss_var_weak}
\end{equation}
Moreover, if there exists regime switching
in the variance of at least one series, i.e., if $\sigma_{1,1}^{2}\neq\sigma_{1,2}^{2}$
or $\sigma_{2,1}^{2}\neq\sigma_{2,2}^{2}$, then the inequality is strict for at least
one $t$ (in fact, already at $t=1$).
\end{corollary}

To supplement the theoretical result, we run a simulation experiment similar to the switching-mean case. We set $\sigma_{1,1}^{2}=\sigma_{1,2}^{2}=\sigma_{2,1}^{2}=1$ and vary $\sigma_{2,2}^{2}$ from 0.1 to 10.
Two extreme cases $\sigma_{2,2}^{2}=0.01$ and $\sigma_{2,2}^{2}=100$ are also included. 
Transition probabilities are the same as in the switching-mean case. Table~\ref{table:switching_variance} reports the MSE and
MAE of filtered regime probabilities. 
\begin{table}[H]
\centering
\caption{MSE and MAE of inference on $S_t$, switching-variance case.}
\renewcommand{\arraystretch}{0.55}
\label{table:switching_variance}
\begin{tabular}{c c c c c c c c c}
\toprule
& \multicolumn{4}{c}{\text{MSE}} 
& \multicolumn{4}{c}{\text{MAE}} \\
\cmidrule(lr){2-5} \cmidrule(lr){6-9}
$\sigma_{2,2}^2$ 
& Dis.
& Agg.
& Imp.
& \% Imp. 
& Dis.
& Agg.
& Imp.
& \% Imp. \\
\midrule
\multicolumn{9}{c}{\text{Benchmark: $\sigma_{2,2}^2=1$}} \\
\midrule
1.00 & 0.1311 & 0.1311 & 0.0000 & 0.0 & 0.2700 & 0.2700 & 0.0000 & 0.0 \\
\midrule
\multicolumn{9}{c}{\text{Panel A: $\sigma_{2,2}^2<1$}} \\
\midrule
0.90 & 0.1309 & 0.1311 & 0.0001 &  0.1 & 0.2696 & 0.2699 & 0.0003 &  0.1 \\
0.80 & 0.1301 & 0.1309 & 0.0008 &  0.6 & 0.2678 & 0.2695 & 0.0017 &  0.6 \\
0.70 & 0.1283 & 0.1306 & 0.0023 &  1.8 & 0.2640 & 0.2689 & 0.0049 &  1.8 \\
0.60 & 0.1252 & 0.1301 & 0.0049 &  3.8 & 0.2574 & 0.2678 & 0.0104 &  3.9 \\
0.50 & 0.1200 & 0.1294 & 0.0095 &  7.3 & 0.2467 & 0.2664 & 0.0196 &  7.4 \\
0.40 & 0.1127 & 0.1283 & 0.0156 & 12.2 & 0.2313 & 0.2641 & 0.0328 & 12.4 \\
0.30 & 0.1021 & 0.1269 & 0.0248 & 19.5 & 0.2093 & 0.2611 & 0.0519 & 19.9 \\
0.20 & 0.0872 & 0.1253 & 0.0381 & 30.4 & 0.1784 & 0.2575 & 0.0792 & 30.8 \\
0.10 & 0.0658 & 0.1232 & 0.0575 & 46.7 & 0.1339 & 0.2529 & 0.1190 & 47.1 \\
0.01 & 0.0253 & 0.1207 & 0.0954 & 79.0 & 0.0512 & 0.2476 & 0.1964 & 79.3 \\
\midrule
\multicolumn{9}{c}{\text{Panel B: $\sigma_{2,2}^2>1$}} \\
\midrule
2.00   & 0.1252 & 0.1289 & 0.0037 &  2.9 & 0.2569 & 0.2650 & 0.0081 &  3.1 \\
3.00   & 0.1176 & 0.1251 & 0.0075 &  6.0 & 0.2411 & 0.2572 & 0.0161 &  6.3 \\
4.00   & 0.1113 & 0.1214 & 0.0101 &  8.3 & 0.2275 & 0.2488 & 0.0214 &  8.6 \\
5.00   & 0.1058 & 0.1178 & 0.0120 & 10.2 & 0.2159 & 0.2410 & 0.0251 & 10.4 \\
6.00   & 0.1012 & 0.1143 & 0.0131 & 11.5 & 0.2066 & 0.2339 & 0.0274 & 11.7 \\
7.00   & 0.0975 & 0.1112 & 0.0138 & 12.4 & 0.1986 & 0.2276 & 0.0289 & 12.7 \\
8.00   & 0.0943 & 0.1086 & 0.0143 & 13.2 & 0.1919 & 0.2217 & 0.0297 & 13.4 \\
9.00   & 0.0914 & 0.1060 & 0.0145 & 13.7 & 0.1858 & 0.2162 & 0.0304 & 14.1 \\
10.00  & 0.0887 & 0.1036 & 0.0149 & 14.4 & 0.1803 & 0.2113 & 0.0310 & 14.7 \\
100.00 & 0.0467 & 0.0566 & 0.0099 & 17.5 & 0.0940 & 0.1141 & 0.0201 & 17.6 \\
\bottomrule
\end{tabular}
\begin{minipage}{1\linewidth}
\footnotesize
\singlespacing
\emph{Notes:} ``Dis.''\ and ``Agg.''\ refer to disaggregated and aggregated data, respectively. ``Imp.''\ denotes absolute improvement of Dis.\ over Agg.; ``\% Imp.''\ denotes percentage improvement and is computed relative to Agg. \end{minipage}
\end{table}

The benchmark case $\sigma_{2,2}^2=1$ corresponds
to no variance switching; the two filters perform identically in terms of MSE and MAE, as predicted by Proposition~2. When $\sigma_{2,2}^2<1$ (Panel~A), the overall variance of the
second series decreases. The difference in MSE and MAE between the two filters grows monotonically as
$\sigma_{2,2}^2$ moves away from one, reaching about 80\% improvement at
$\sigma_{2,2}^2=0.01$. A similar pattern is seen under $\sigma_{2,2}^2>1$ (Panel~B), although the gains are
more modest. In this case, the overall variance of the second series
increases, which spreads out probability mass and flattens the likelihood, making regime switching harder to detect. In summary, the results confirm that aggregation loses no information only in the trivial case of no switching; otherwise the disaggregated filter is strictly more precise. In addition, switches into a lower-variance regime are easier to detect, and aggregation costs more in that case.

\section{Algorithm for inference and forecasting}

Let $\theta$ denote the vector of model parameters, including the means $\mu_{j,i,s}$, variances $\sigma^{2}_{j,i,s}$, autoregressive coefficients $\phi_{j,m}$, and transition probabilities $p_{ss'}$, where $j$ indexes groups, $i$ indexes series within group $j$, $m$ denotes the lag order, and $s$ indexes regimes. The dimension of $\theta$ is
\(
M = 2S\sum_{j=1}^{J} N_j + kJ + S(S-1),
\)
which exceeds 200 when there are 50 series. In such settings, likelihood-based optimization often fails to converge. Instead, we use a Bayesian approach that samples $\theta$ and the latent state sequence $\{S_t\}_{t=1}^T$ jointly. The exposition focuses on the disaggregate specification in \eqref{eq:2.1}; differences that arise under \eqref{eq:2.2} will be noted where relevant.

Our sampler adapts the sampling algorithms in
\cite{albert1993}, \cite{kim1998}, and \cite{kim1999} to a high-dimensional
setting. The goal is to sample from the following
joint posterior distribution conditional on
$y_{1:T}=(y_{1}^{\prime},y_{2}^{\prime},\ldots,y_{T}^{\prime})^{\prime}$:
\begin{align*}
 f\!\left(
S_{1:T},
\theta,
y_{T+1},\ldots,y_{T+l},
S_{T+1},\ldots,S_{T+l}
\mid
y_{1:T}
\right),
\end{align*}
where $S_{1:T}$ is the latent state sequence over the sample, $l$
is the forecast horizon, and $y_{T+1},\ldots,y_{T+l}$ and
$S_{T+1},\ldots,S_{T+l}$ are future observables and latent states,
respectively. As often done in the literature, we factor this joint distribution as a conditional and a marginal distribution:
\begin{align*}
 f\!\left(
y_{T+1},\ldots,y_{T+l},
S_{T+1},\ldots,S_{T+l}
\mid
y_{1:T},S_{1:T},
\theta
\right) f\!\left(
S_{1:T},
\theta
\mid
y_{1:T}
\right).
\end{align*}
Once the within-sample latent states and parameter values have been
sampled, future states and observables can be generated conditional on these
draws. For the distribution $f\!\left(S_{1:T},\theta \mid y_{1:T}\right)$, we consider its full set of conditionals and sample from them iteratively:
\begin{align*}
\makebox[8.2cm][l]{\textup{Latent states:}}
& f\!\left(
S_{1:T}
\mid
y_{1:T},
\mu_{j,i,s},
\sigma_{j,i,s}^{2},
\phi_{j,m},
p_{ss'}
\right), \\
\makebox[8.2cm][l]{\textup{Regression coefficients and variances:}}
& f\!\left(
\mu_{j,i,s},
\sigma_{j,i,s}^{2},
\phi_{j,m}
\mid y_{1:T},
S_{1:T},
p_{ss'}
\right), \\
\makebox[8.2cm][l]{\textup{Transition matrix:}}
& f\!\left(
p_{ss'}
\mid
y_{1:T},
S_{1:T},
\mu_{j,i,s},
\sigma_{j,i,s}^{2},
\phi_{j,m}
\right).
\end{align*}
Here $\mu_{j,i,s}$, $\sigma_{j,i,s}^2$, and
$\phi_{j,m}$ denote the full parameter
collections $\{\mu_{j,i,s}\}_{j,i,s}$, $\{\sigma_{j,i,s}^2\}_{j,i,s}$, and
$\{\phi_{j,m}\}_{j,m}$. 

Next we provide details for each step, with an emphasis on the aspects related to the high-dimensional setting. We assume there are two
states, $S_t \in \{1,2\}$. Note that for a given \(t\), the log likelihood conditional on \(S_{t-k:t}\) is
\begin{equation}
\log p\!\left(y_t \mid S_{t-k:t},\, y_{1:t-1},\,\theta\right)
=
-\frac{1}{2}
\sum_{j=1}^{J}\sum_{i=1}^{N_j}
\left[
\log\!\bigl(2\pi\sigma_{j,i,S_t}^2\bigr)
+
\frac{\varepsilon_{j,i,t}^2}{\sigma_{j,i,S_t}^2}
\right],
\label{eq:log_obs_density_A}
\end{equation}
where
$
\varepsilon_{j,i,t}
=
y_{j,i,t}
-
\mu_{j,i,S_t}
-
\sum_{m=1}^{k}
\phi_{j,m}
\bigl(y_{j,i,t-m}-\mu_{j,i,S_{t-m}}\bigr)$
for $t=k+1,\ldots,T .
$
Summing over \(t=k+1,\ldots,T\), we get the log likelihood for the full sample:
\begin{equation*}
\log p\!\left(y_{k+1:T} \mid S_{1:T},\,y_{1:k},\,\theta\right)
=
-\frac{1}{2}
\sum_{t=k+1}^{T}
\sum_{j=1}^{J}\sum_{i=1}^{N_j}
\left[
\log\!\bigl(2\pi\sigma_{j,i,S_t}^2\bigr)
+
\frac{\varepsilon_{j,i,t}^2}{\sigma_{j,i,S_t}^2}
\right].
\end{equation*}

\subsection{Step 1: Generate regression coefficients and variances}

In each iteration, we first sample $\mu_{j,i,s}$, followed by $\sigma_{j,i,s}^{2}$, and then $\phi_{j,m}$. 

\subsubsection{Generate the means}

The regime-dependent means are sampled from $ f\!(
\mu_{j,i,s}
\mid
y_{1:T},
S_{1:T},
p_{ss'},\sigma_{j,i,s}^{2},
\phi_{j,m})$ at once. Because the sampling is conditional on \(\{\phi_{j,m}\}_{m=1}^k\) and \(S_{1:T}\), Model~A can be rewritten as
\[
y_{j,i,t}
-
\sum_{m=1}^{k}\phi_{j,m}y_{j,i,t-m}
=
\delta_{j,i,1}
\bigl(
\mathbf{1}_{\{S_{t}=1\}}
-
\sum_{m=1}^{k}\phi_{j,m}\mathbf{1}_{\{S_{t-m}=1\}}
\bigr)
+
\mu_{j,i,2}
\bigl(
1-\sum_{m=1}^{k}\phi_{j,m}
\bigr)
+
e_{j,i,t}.
\]
For each \(t\), this represents a system of \(N\) seemingly unrelated regression equations, with outcome
\(y_{j,i,t}-\sum_{m=1}^{k}\phi_{j,m}y_{j,i,t-m}\),
regressors
\(
\mathbf{1}_{\{S_{t}=1\}}
-
\sum_{m=1}^{k}\phi_{j,m}\mathbf{1}_{\{S_{t-m}=1\}}
\quad\text{and}\quad
1-\sum_{m=1}^{k}\phi_{j,m},
\)
coefficients \(\delta_{j,i,1} = \mu_{j,i,1} - \mu_{j,i,2}\le 0\) and \(\mu_{j,i,2}\), and normal error terms with a diagonal covariance matrix. 
Let $\delta_{1}$ be a vector containing all $\delta_{j,i,1}$, and
$\mu_{2}$ all \(\mu_{j,i,2}\).
Under the truncated normal priors that we use, the posterior distribution of \(\delta_{1}\) is a truncated multivariate normal, and that of \(\mu_{2}\) is a standard multivariate normal (see \citet{Zellner1971}). Our sampling relies on vectorized operations and exploits independence across equations, so that the required matrix inversion has dimension at most \(2N\), regardless of the number of observations in the system. Sampling for the truncated normal uses the minimax tilting method of \cite{Botev2017}, specifically designed for high-dimensional data. Details on these aspects can be found in the Appendix. Monte Carlo experiments show that drawing \(\delta_{1}\) and \(\mu_{2}\) jointly is feasible for up to a few hundred time series. When the cross-sectional dimension is even larger, the elements of these vectors can be sampled in blocks, provided that the priors are independent across parameters.


\subsubsection{Generate the variances}
The variances are sampled from $f\!(
\sigma_{j,i,s}^{2}
\mid
y_{1:T},
S_{1:T},
p_{ss'},\mu_{j,i,s},
\phi_{j,m}).$ The procedure uses independence across time series and samples the variances one at a time. The explicit sampling steps are identical to those in \citet[pp.~216--218]{kim1999}. The variances are expressed as
\begin{equation*}
\sigma_{j,i,S_t}^{2}
=
\sigma_{j,i,2}^{2}
\left[
1+\left(\frac{\sigma_{j,i,1}^{2}}{\sigma_{j,i,2}^{2}}-1\right)\mathbf{1}_{\{S_t=1\}}
\right]
=
\sigma_{j,i,2}^{2}\left[1+h_{j,i,2}\mathbf{1}_{\{S_t=1\}}\right],
\end{equation*}
where \(h_{j,i,2}=\sigma_{j,i,1}^{2}/\sigma_{j,i,2}^{2}-1\). With independent
inverse-Gamma priors on \(\sigma_{j,i,2}^{2}\) and \(1+h_{j,i,2}\), the conditional
posterior distributions are also inverse-Gamma and can be sampled directly. Each step is
a one-dimensional problem and does not require any matrix inversion. 
Implementation details are included in the Appendix.

\subsubsection{Generate the autoregressive coefficients}
The autoregressive coefficients are sampled from
\(
f\!(
\phi_{j,m}
\mid
y_{1:T},
S_{1:T},
p_{ss'},
\mu_{j,i,s},
\sigma_{j,i,s}^{2}).
\) These are drawn separately from $\mu_{j,i,s}$ to avoid nonlinear interactions between them. Following the same approach as for $\mu_{j,i,s}$, we write Model~A as
\[
y_{j,i,t}-\mu_{j,i,S_t}
=
\sum_{m=1}^{k}\phi_{j,m}\bigl(y_{j,i,t-m}-\mu_{j,i,S_{t-m}}\bigr)
+
e_{j,i,t},
\qquad
e_{j,i,t}\mid S_t \sim N(0,\sigma_{j,i,S_t}^2).
\]
Stacking across groups gives a system with dependent variables $y_{j,i,t}-\mu_{j,i,S_t}$, regressors $y_{j,i,t-m}-\mu_{j,i,S_{t-m}}$, common coefficients $\phi_{j,m}$ within group, and diagonal residual covariance. Under a normal prior on $\phi_{j,m}$, the conditional
posterior distribution is normal and can be sampled directly. As in the case of sampling $\mu$, we adopt a vectorized representation so that
the required matrix inversion has dimension at most $Jk$, regardless of the
cross-sectional dimension. Implementation details are provided in the
Appendix.
\subsection{Step 2: Generate latent states} 
The latent states are sampled from $ f\!\left(
S_{1:T}
\mid
y_{1:T},
\theta\right)$.  We build on the multi-move Gibbs sampler of \citet{kim1998}, which generates the entire sequence of latent states at once. It does so by factoring the joint posterior distribution into conditionals and sampling backward in time:
\[
f\!\left(S_{1:T}\mid y_{1:T},\theta\right)
=
f\!\left(S_T\mid y_{1:T},\theta\right)
\prod_{t=1}^{T-1}
f\!\left(S_t \mid S_{t+1:T},\, y_{1:T},\,\theta\right).
\]
Let \(k\ge 0\) denote the number of state lags entering the observation density. In the Appendix, we show that the following relationship between smoothing and filtering holds for Model A:
\begin{equation}
p\!\left(S_t \mid S_{t+1:T},\,  y_{1:T},\, \theta\right)
\;\propto\;
p\!\left(S_{t:t+k}\mid  y_{1:t+k},\, \theta\right)
\label{sampler-2}
\end{equation}
for any \(k>0\) and $p\!\left(S_t \mid S_{t+1:T},\,  y_{1:T},\, \theta\right)
\;\propto\;
p\!\left(S_{t:t+1}\mid  y_{1:t+1},\, \theta\right)$ for $k=0.$
This delivers an exact sampler for the smoothed distribution (left-hand side) as a by-product of the state filtering step in the filter of \citet{hamilton1989} (right-hand side). The sampler of \cite{kim1998} is exact only if \(k=0\); when \(k>0\), it is approximate because the dependence on some states is omitted.

In implementation, the right hand side of (\ref{sampler-2}) is computed forward recursively as follows. Given \(p(S_{t-k-1:t-1}\mid y_{1:t-1},\theta)\), the Markov property of \(\{S_t\}\) implies the following prediction:
\begin{equation}
p\!\left(S_{t-k:t}\mid y_{1:t-1},\,\theta\right)
=
p_{S_{t-1},S_t}
\sum_{S_{t-k-1}\in\mathcal S}
p\!\left(S_{t-k-1:t-1}\mid  y_{1:t-1},\,\theta\right).
\label{eq:prediction_A}
\end{equation}
Then, given $y_t$, the state probability can be updated as:
\begin{equation*}
p\!\left(S_{t-k:t}\mid  y_{1:t},\,\theta\right)
\;\propto\;
p\!\left(y_t \mid S_{t-k:t},\, y_{1:t-1},\,\theta\right)
p\!\left(S_{t-k:t}\mid  y_{1:t-1},\,\theta\right).
\end{equation*}
This update is carried out using log weights because the observation density involves products over many series, which can become numerically small. Specifically, for each configuration \((S_{t-k},\ldots,S_t)\), we compute the unnormalized log weight as the sum of the log observation density in \eqref{eq:log_obs_density_A} and the log predictive probability from \eqref{eq:prediction_A}. To improve numerical stability, the maximum log weight across all configurations is subtracted before exponentiation, after which the resulting weights are normalized to sum to one.  
Given the filtered joint distributions
\(\{p(S_{t-k:t}\mid y_{1:t},\theta)\}\),
sampling proceeds backward starting at \(T\). At time \(t\), one draws from
\(
p\!\left(S_t,\, S_{t+1:t+k}\mid y_{1:t+k},\, \theta\right),
\)
which is obtained directly from the filtered probabilities at time \(t+k\).
When \(t+k>T\), the conditioning set \(S_{t+1:t+k}\) is truncated accordingly.  These operations do not involve any matrix operations, and the sampling involves
drawing Bernoulli random variables. The procedure remains feasible for large cross-sections. To our knowledge, this exact state sampler is new.


\subsection{Step 3: Generate the transition matrix}

The elements of the transition matrix  are sampled from $ f\!\left(
p_{ss'}
\mid
y_{1:T},
S_{1:T},
\mu_{j,i,s},
\sigma_{j,i,s}^{2},
\phi_{j,m}
\right)= f\!\left(
p_{ss'}
\mid
S_{1:T}
\right)$, as in \citet{kim1998}. Since \(S_t\) is a scalar process, the sampling is independent of the model's dimension. We use independent Beta priors for \(p_{11}\) and \(p_{22}\), which implies the posterior distributions are also Beta. The Appendix provides the exact expressions.

\subsection{Step 4: Forecast future observables}

The values of $y_{T+1:T+l}$ for any $l\geq1$ are sampled from
$
f(y_{T+1:T+l}, S_{T+1:T+l}\mid y_{1:T}, S_{1:T}, \theta).
$
Since
$
f(y_{T+1:T+l}, S_{T+1:T+l}\mid y_{1:T}, S_{1:T}, \theta)
=
\prod_{j=1}^{l}
f(y_{T+j}, S_{T+j}\mid y_{1:T+j-1}, S_{1:T+j-1}, \theta),
$
the joint distribution can be sampled sequentially from the $l$ conditional distributions on the right-hand side. For each $j=1,\ldots,l$, we draw $S_{T+j}$ from the transition distribution
$
f(S_{T+j}\mid S_{T+j-1}, \theta),
$
and then draw $y_{T+j}$ from
$
f(y_{T+j}\mid y_{1:T+j-1}, S_{1:T+j}, \theta).
$
Repeating this procedure for $j=1,\ldots,l$ gives a draw from
$
f(y_{T+1:T+l}, S_{T+1:T+l}\mid y_{1:T}, S_{1:T}, \theta).
$
Applying this step to each posterior draw of $(\theta, S_{1:T})$ produces predictive simulations of future observations. 

This part is essentially the same as in low-dimensional regime-switching models. It is tractable in the high-dimensional case because the state $S_{T+j}$ is common across all series. Only one draw is needed for the state, and the observations are then drawn independently across series.

\subsection{Prior}

The following prior distributions are used for the simulation experiments and the empirical application. The priors are assumed to be independent unless stated otherwise.

\paragraph{Regime-dependent means.}
For each series, the regime-dependent means are assigned priors
$
\delta_{j,i,1}=\mu_{j,i,1}-\mu_{j,i,2} \sim N(-0.5,\,50^2)$ and 
$\mu_{j,i,2} \sim N(0,\,50^2).
$
The prior variances are set to be diffuse, reflecting the growth rates can differ substantially across disaggregated series.

\paragraph{Innovation variances.}

When the variances are restricted to be constant across regimes, we set
$\sigma_{j,i}^2 \sim IG(0,0)$, which corresponds to the improper prior
$p(\sigma_{j,i}^2)\propto (\sigma_{j,i}^2)^{-1}$. When the variances are allowed to switch across regimes, we still assign $\sigma_{j,i,2}^2 \sim IG(0,0)$, and as in \citet{kim1999}, we write $\sigma_{j,i,1}^2 = (1+h_{j,i,2})\sigma_{j,i,2}^2$ and adopt a conjugate inverse-gamma prior for $1+h_{j,i,2}$. We set  $1+h_{j,i,2} \sim IG\!\left(\tfrac{T_{1}}{2},\,\tfrac{T_{1}+2}{2}\right)$, where $T_{1}$ denotes the number of observations assigned to regime 1 under the current draw of $S_{1:T}$. This represents a state-dependent prior rather than a fixed prior. Under this specification, the mode of $1+h_{j,i,2}$ is equal to 1 (the prior for the recession variance peaks at the expansion variance), and the amount of shrinkage increases with $T_{1}$. The shrinkage is helpful because the recession regime typically contains few observations. The prior also allows for right-skewed dispersion in the variance ratio.

\paragraph{AR coefficients.}
For $k>0$, the AR coefficients (by group) are assigned Gaussian prior
\(
N(0,\,0.5^2), 
\)
with the stationarity restriction enforced by rejection sampling, i.e., draws are accepted
only if the roots of the AR polynomial lie outside the unit circle. 

\paragraph{Transition probabilities.}
For the Markov chain, the diagonal transition probabilities
$p_{11}$ (recession) and $p_{22}$ (expansion) are assigned priors
$
p_{11}\sim \mathrm{Beta}(2,2)$ and
$p_{22}\sim \mathrm{Beta}(30,2),
$
with means $0.50$ and $0.94$, and standard deviations $0.22$ and $0.04$. 
These values are informed by previous estimates of regime-switching models 
for U.S. business cycles. The tight prior on $p_{22}$ reflects the consistent  empirical evidence and reinforces the separation between regimes.

\section{Monte Carlo experiments}

This section examines finite-sample effects of disaggregation on state recovery. We begin with a two-variable setting in which the signal-to-noise ratio varies and compare aggregated and disaggregated filters. This analysis complements the asymptotic characterization in Section 2.2 by incorporating parameter estimation uncertainty. We then consider a setting where the signal-to-noise ratio is fixed while the cross-sectional dimension increases, and examine how filtering and smoothing change as additional series are included. The sample size set to 190, matching the empirical sample size used to estimate model (\ref{eq:MS}). Parameters are estimated via MCMC using the priors in Section 3.5. All results are averages over 100 independent replications.

\subsection{Setting 1: two variables, varying signal-to-noise ratio}\label{sec:4.1}
We use the DGP in (\ref{eq:mean_switch_model}) with the same parameter values. The estimated model has the same specification as the DGP, with no lags, and assumes constant variances across regimes. We report filtered probabilities, as those in Table \ref{table:agg_mean_only}; results based on smoothed probabilities are similar.

Table \ref{table:switching_mean} shows that, relative to Table \ref{table:agg_mean_only}, where parameters are known, the MSE and MAE are generally larger for both filters, reflecting the uncertainty from parameter estimation. At the benchmark $\sigma_2^2 = 1$, the differences in MSE and MAE between the aggregated and disaggregated filters are negligible, in line with the theory. Gains from disaggregation emerge as $\sigma_2^2$ moves away from 1 in either direction. For $\sigma_2^2 < 1$, the improvement reaches about $50\%$ at $\sigma_2^2 = 0.3$ and close to $100\%$ for $\sigma_2^2 \leq 0.10$. For $\sigma_2^2 > 1$, gains grow more slowly but are still substantial -- about $30\%$ by $\sigma_2^2 = 4$ and above $45\%$ at $\sigma_2^2 = 10$. The monotonic pattern seen in Table \ref{table:agg_mean_only} is preserved under parameter estimation uncertainty.

Therefore, although parameter uncertainty increases the MSE and MAE for both filters, it does not alter the main finding: aggregation attenuate regime signals, and the gains from disaggregation can be large when heterogeneity is substantial.

\begin{table}[H]
\centering
\caption{Finite-Sample MSE and MAE of inference on $S_t$, switching-mean case.}
\renewcommand{\arraystretch}{0.55}
\label{table:switching_mean}
\begin{tabular}{c c c c c c c c c}
\toprule
& \multicolumn{4}{c}{\text{MSE}} 
& \multicolumn{4}{c}{\text{MAE}} \\
\cmidrule(lr){2-5} \cmidrule(lr){6-9}
$\sigma_{2}^{2}$ 
& Dis.
& Agg.
& Imp.
& \% Imp. 
& Dis.
& Agg.
& Imp.
& \% Imp. \\
\midrule
\multicolumn{9}{c}{\text{Benchmark: $\sigma_{2}^{2}=1$}} \\
\midrule
1.00 & 0.0278 & 0.0280 & 0.0002 & 0.7 
     & 0.0570 & 0.0585 & 0.0015 & 2.6 \\
\midrule
\multicolumn{9}{c}{\text{Panel A: $\sigma_{2}^{2}<1$}} \\
\midrule
0.90 & 0.0252 & 0.0255 & 0.0003 & 1.2 
     & 0.0528 & 0.0544 & 0.0016 & 2.9 \\
0.80 & 0.0239 & 0.0245 & 0.0006 & 2.4 
     & 0.0496 & 0.0518 & 0.0022 & 4.2 \\
0.70 & 0.0205 & 0.0217 & 0.0012 & 5.5 
     & 0.0415 & 0.0450 & 0.0035 & 7.8 \\
0.60 & 0.0184 & 0.0202 & 0.0018 & 8.9 
     & 0.0368 & 0.0418 & 0.0050 & 12.0 \\
0.50 & 0.0139 & 0.0183 & 0.0044 & 24.0 
     & 0.0301 & 0.0401 & 0.0100 & 24.9 \\
0.40 & 0.0107 & 0.0159 & 0.0052 & 32.7 
     & 0.0224 & 0.0337 & 0.0113 & 33.5 \\
0.30 & 0.0058 & 0.0133 & 0.0075 & 56.4 
     & 0.0130 & 0.0287 & 0.0157 & 54.7 \\
0.20 & 0.0025 & 0.0111 & 0.0086 & 77.5 
     & 0.0052 & 0.0247 & 0.0195 & 78.9 \\
0.10 & 0.0002 & 0.0100 & 0.0098 & 98.0 
     & 0.0006 & 0.0220 & 0.0214 & 97.3 \\
0.01 & 0.0000 & 0.0076 & 0.0076 & 100.0 
     & 0.0000 & 0.0173 & 0.0173 & 100.0 \\
\midrule
\multicolumn{9}{c}{\text{Panel B: $\sigma_{2}^{2}>1$}} \\
\midrule
2.00   & 0.0415 & 0.0454 & 0.0039 & 8.6 
       & 0.0826 & 0.0949 & 0.0123 & 13.0 \\
3.00   & 0.0418 & 0.0547 & 0.0129 & 23.6 
       & 0.0864 & 0.1139 & 0.0275 & 24.1 \\
4.00   & 0.0465 & 0.0691 & 0.0226 & 32.7 
       & 0.0967 & 0.1423 & 0.0456 & 32.0 \\
5.00   & 0.0476 & 0.0742 & 0.0266 & 35.8 
       & 0.0973 & 0.1518 & 0.0545 & 35.9 \\
6.00   & 0.0508 & 0.0842 & 0.0334 & 39.7 
       & 0.1002 & 0.1705 & 0.0703 & 41.2 \\
7.00   & 0.0504 & 0.0866 & 0.0362 & 41.8 
       & 0.1010 & 0.1762 & 0.0752 & 42.7 \\
8.00   & 0.0527 & 0.0965 & 0.0438 & 45.4 
       & 0.1054 & 0.1900 & 0.0846 & 44.5 \\
9.00   & 0.0521 & 0.0982 & 0.0461 & 46.9 
       & 0.1068 & 0.1938 & 0.0870 & 44.9 \\
10.00  & 0.0524 & 0.1058 & 0.0534 & 50.5 
       & 0.1065 & 0.1975 & 0.0910 & 46.1 \\
100.00 & 0.0601 & 0.1456 & 0.0855 & 58.7 
       & 0.1150 & 0.2363 & 0.1213 & 51.3 \\
\bottomrule
\end{tabular}
\begin{minipage}{1\linewidth}
\footnotesize
\singlespacing
\emph{Notes:} 
``Dis.''\ and ``Agg.''\ refer to disaggregated and aggregated data, respectively. 
``Imp.''\ denotes absolute improvement of Dis.\ over Agg.; ``\% Imp.''\ denotes percentage improvements and is computed relative to Agg.
\end{minipage}
\end{table}

\subsection{Setting 2: Increasing $N$ with a fixed signal-to-noise ratio}\label{sec:4.2}

We calibrate the DGP to empirical estimates from model (\ref{eq:MS}). These estimates provide 30 sets of regime-specific means and variances, along with a common transition matrix. For simulation, we first generate a common latent state vector of length 190 using this transition matrix. We then draw one of the 30 parameter sets at random, scale the regime difference $\delta_{i,1} = \mu_{i,1} - \mu_{i,2}$ by one-half, and use it together with the latent state vector to generate a time series of length 190 from model (\ref{eq:MS}). This procedure is repeated to obtain $N$ series, with $N \in \{5,10,20,40,60,80,100,200,400\}$. The estimated model matches the specification of (\ref{eq:MS}), with no lags and variances allowed to switch across regimes.

Here the regime difference $\delta_{i,1}$ is multiplied by 0.5 to keep regime classification nontrivial even in large cross sections. Otherwise, the advantage of disaggregation remains, but the regimes are estimated with little uncertainty in both the aggregated and disaggregated models once $N$ exceeds 60, making this experiment uninformative about relative performance at higher dimensions.

\begin{table}[H]
\centering
\caption{MSE and MAE of inference on $S_t$, DGP calibrated to GDP components.}
\label{table:m6_mae_mse}
\renewcommand{\arraystretch}{0.55}
\begin{tabular}{c c c c c c c c c}
\toprule
& \multicolumn{4}{c}{\text{MSE}} 
& \multicolumn{4}{c}{\text{MAE}} \\
\cmidrule(lr){2-5} \cmidrule(lr){6-9}
$N$ 
& Dis.
& Agg.
& Imp.
& \% Imp. 
& Dis.
& Agg.
& Imp.
& \% Imp. \\
\midrule
\multicolumn{9}{c}{\text{Panel A: Filtered}} \\
\midrule
5   & 0.1278 & 0.1350 & 0.0072 & 5.3 & 0.1560 & 0.1972 & 0.0412 & 20.9 \\
10  & 0.0872 & 0.1300 & 0.0428 & 32.9 & 0.1230 & 0.1889 & 0.0659 & 34.9 \\
20  & 0.0416 & 0.1171 & 0.0755 & 64.5 & 0.0695 & 0.1711 & 0.1016 & 59.4 \\
40  & 0.0220 & 0.0957 & 0.0737 & 77.0 & 0.0374 & 0.1486 & 0.1112 & 74.8 \\
60  & 0.0106 & 0.0676 & 0.0570 & 84.3 & 0.0178 & 0.1141 & 0.0963 & 84.4 \\
80  & 0.0056 & 0.0486 & 0.0430 & 88.5 & 0.0092 & 0.0860 & 0.0768 & 89.3 \\
100 & 0.0023 & 0.0346 & 0.0323 & 93.4 & 0.0038 & 0.0631 & 0.0593 & 94.0 \\
200 & 0.0003 & 0.0085 & 0.0082 & 96.5 & 0.0004 & 0.0167 & 0.0163 & 97.6 \\
400 & 0.0000 & 0.0018 & 0.0018 & 100.0 & 0.0001 & 0.0038 & 0.0037 & 97.4 \\
\midrule
\multicolumn{9}{c}{\text{Panel B: Smoothed}} \\
\midrule
5   & 0.1236 & 0.1319 & 0.0083 & 6.3 & 0.1475 & 0.1916 & 0.0441 & 23.0 \\
10  & 0.0817 & 0.1264 & 0.0447 & 35.4 & 0.1101 & 0.1823 & 0.0722 & 39.6 \\
20  & 0.0366 & 0.1117 & 0.0751 & 67.2 & 0.0577 & 0.1622 & 0.1045 & 64.4 \\
40  & 0.0196 & 0.0879 & 0.0683 & 77.7 & 0.0310 & 0.1354 & 0.1044 & 77.1 \\
60  & 0.0097 & 0.0578 & 0.0481 & 83.2 & 0.0149 & 0.0968 & 0.0819 & 84.6 \\
80  & 0.0050 & 0.0396 & 0.0346 & 87.4 & 0.0078 & 0.0695 & 0.0617 & 88.8 \\
100 & 0.0021 & 0.0273 & 0.0252 & 92.3 & 0.0033 & 0.0491 & 0.0458 & 93.3 \\
200 & 0.0003 & 0.0056 & 0.0053 & 94.6 & 0.0004 & 0.0116 & 0.0112 & 96.6 \\
400 & 0.0000 & 0.0012 & 0.0012 & 100.0 & 0.0001 & 0.0025 & 0.0024 & 96.0 \\
\bottomrule
\end{tabular}
\vspace{0.5em}
\begin{minipage}{1\linewidth}
\singlespacing
\footnotesize
\emph{Notes:} MSE and MAE denote mean squared and mean absolute error. $N$ is the number of series in disaggregated data.
``Dis.''\ and ``Agg.''\ refer to disaggregated and aggregated data, respectively. 
``Imp.''\ denotes absolute improvement of Dis.\ over Agg.; ``\% Imp.''\ denotes percentage improvements and is computed relative to Agg.
\end{minipage}
\end{table}
Table \ref{table:m6_mae_mse} reports the MSE and MAE of inference on $S_t$, with panel A for filtering and panel B for smoothing. For filtering, the two models perform comparably at $N = 5$. With $N = 10$ series, disaggregation reduces the filtered MSE from 0.1300 to 0.0872, an improvement of about 32\%. At $N = 20$, the filtered MSE is 0.1171 for the aggregated model and 0.0416 for the disaggregated model, a reduction of about 65\%. At $N = 100$, the filtered MSE is 0.0346 for the aggregated model and 0.0023 for the disaggregated model, a reduction of about 93\%. Smoothing shows similar percentage reductions. As $N$ further increases from 100, MSEs for both approaches become small. This reflects the design: since the signal-to-noise ratio of each series is held fixed, the total information increases with $N$. These results suggest that disaggregation improves state inference across a wide range of cross-sectional sample sizes in this empirically calibrated setting.

\begin{remark}For computational cost, when run on an Intel Xeon CPU (2.80 GHz, using 8 cores), the time to produce a total of 50K draws is as follows: 5.2 minutes for M = 60, 13.5 minutes for M = 100, 1.0 hour for M = 200, and 5.2 hours for M = 400. This shows that the procedure remains practical for M up to 400, which covers the majority of macroeconomic applications. We further document computational times in the application section.
\end{remark}

\section{Empirical application} \label{sec:application}

The NBER identifies US expansions and recessions based on a broad set of economic indicators, and its dating decisions are not tied to any single econometric model. Markov regime-switching model is arguably the most important
econometric model for identifying business cycles. It provides a
probabilistic statement about the state of the economy, including the
likelihood of a recession. Among existing studies, the number of time series
analyzed is often limited compared with the NBER's information set. Also, in most
cases, only aggregate variables are included in the analysis. This section
departs from the existing literature by applying the proposed models to disaggregated time series. As shown below, this approach yields sharper inference on the latent states than is typically available in the existing literature.

\subsection{Data structure}

This subsection describes the data and previews the model. We use five sets of disaggregated series: components of real GDP, industrial production by industry, capacity utilization by industry, nonfarm payroll employment by sector, and average weekly hours by sector. There are 94 series in total. The detailed sectoral and industry breakdowns are reported in the Appendix. We use seasonally adjusted quarterly data. Log differences are used to obtain percentage changes. 

Because some disaggregate series begin in 1972 (e.g., capacity utilization for wood products and for computer and electronic products), the sample spans 1972–2019. We first exclude the COVID period so that it does not dominate the estimation. We then extend the sample to include it, using winsorization to limit its influence on estimation (more details are given below).
\begin{figure}[H]
\centering
\caption{Cross Sectional Distribution of Disaggregate Series}
\vspace{-0.3cm} \includegraphics[width=0.8\linewidth, trim=0.9cm 5.5cm 1cm 10cm, clip]{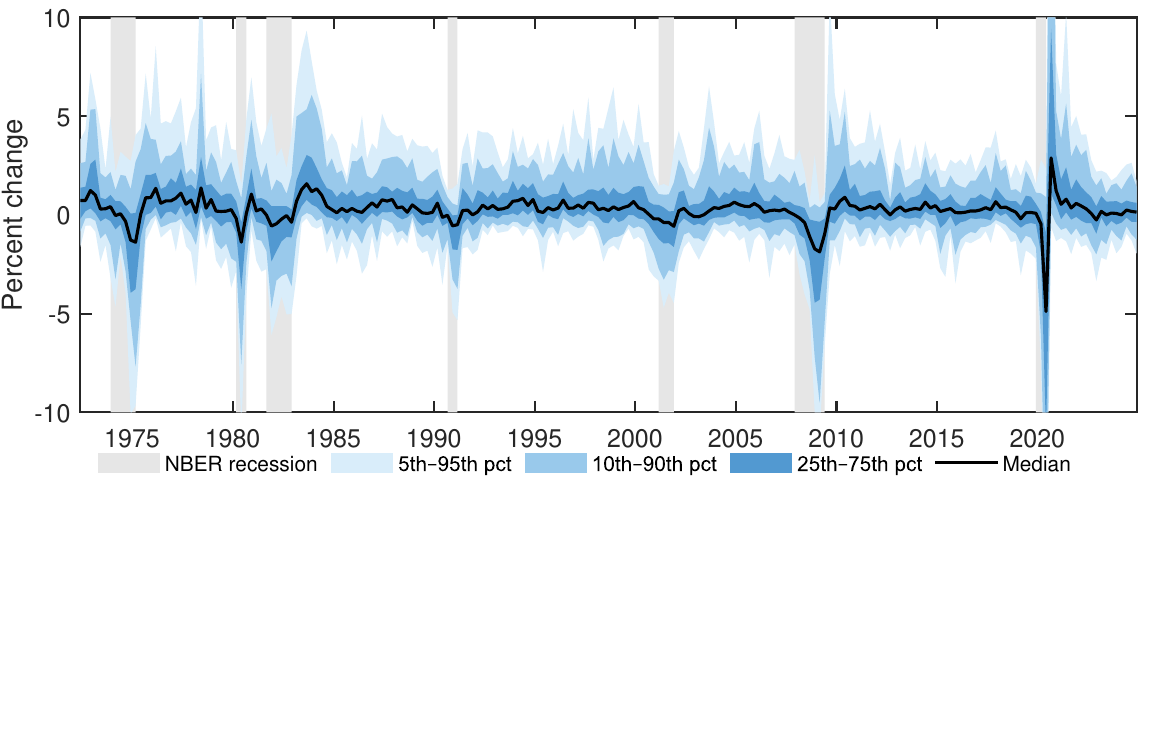}
\label{fig:Fanchart}
\vspace{-3.0cm}
\begin{minipage}{0.8\linewidth}
\singlespacing
\footnotesize\textit{Note:} The figure displays the cross-sectional distribution of 94 disaggregate series. Shaded areas represent the 5th--95th and 25th--75th percentile bands. The solid line denotes the cross-sectional median. Grey shaded regions indicate NBER recession periods.
\end{minipage}
\end{figure}

For illustration, Figure~\ref{fig:Fanchart} plots the cross-sectional distribution of the disaggregated series as a fan chart at each point in time, with shaded areas corresponding to NBER recessions. The series exhibit clear comovement, as shown by the quantiles relative to the median, indicating the presence of common cycles. The variability of the series is noticeably higher during NBER recessions than during expansions, and is particularly high during severe recessions. This confirms the importance of allowing the variances of the series to switch across regimes. From an identification viewpoint, it implies that changes in variances contain useful information for identifying business cycles, a point that is not always recognized in the literature.

We first estimate models without autoregressive components (Models A and B with lag order set to zero), first on the pre-COVID sample (1972:II–2019:III) and then on the full sample through 2024:IV. Then we consider models with lagged dependent variables, with the lag order set to one, and data divided into two groups: GDP components and the rest. The variance is always allowed to switch. As a benchmark, Figure~\ref{fig:Fred_recession} reports smoothed recession probabilities from FRED (series RECPROUSM156N) based on the model of Chauvet (1998), which is representative of estimates typically reported in the literature. These probabilities are at the quarterly frequency, constructed from monthly data using the end-of-period method.

\subsection{Model without autoregressive components}

Figure \ref{fig:application_nolag}(a) shows the smoothed recession probabilities together with NBER recession dates (peak to trough quarter) for the 1972-2019 sample. The model detects all NBER-defined recessions. With very few exceptions, the probabilities are either zero or one, showing high posterior certainty about the prevailing regime. This contrasts with the probabilities reported by the FRED series, which can linger at intermediate values for multiple periods. This difference confirms the informativeness of the 94 disaggregate series in identifying the latent state.
\begin{figure}[H]
    \centering

    \caption{Smoothed Recession Probability without Autoregressive Dynamics}
    \label{fig:application_nolag}
    \begin{minipage}[t]{0.5\textwidth}
        \centering
        {\small (a) 1972-2019}\\
        \includegraphics[width=\linewidth, clip, trim=2cm 9cm 0cm 10cm]{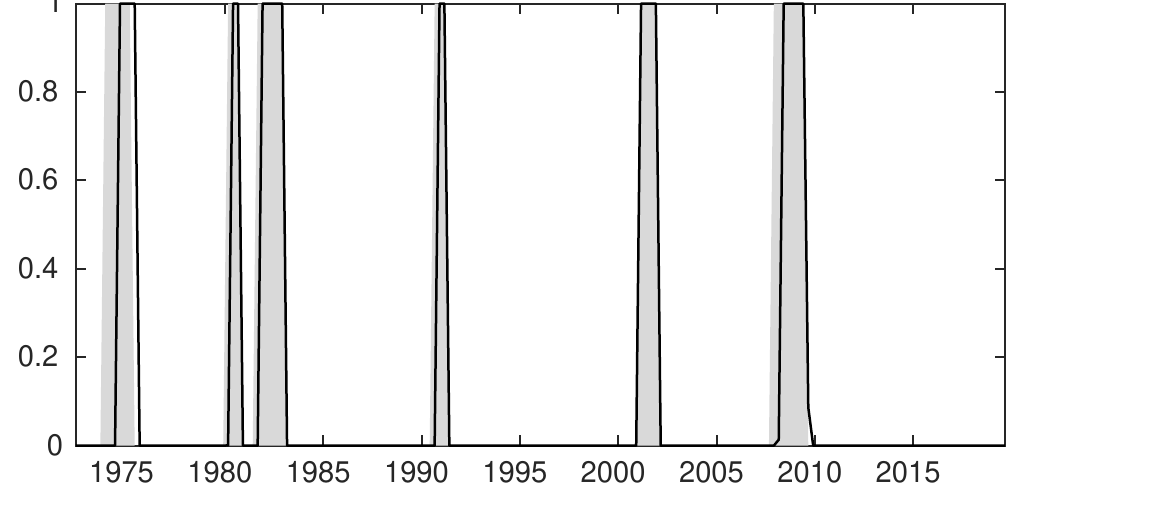}
    \end{minipage}%
    \begin{minipage}[t]{0.5\textwidth}
        \centering
        {\small (b) 1972-2024}\\
        \includegraphics[width=\linewidth, clip, trim=2cm 9cm 1cm 10cm]{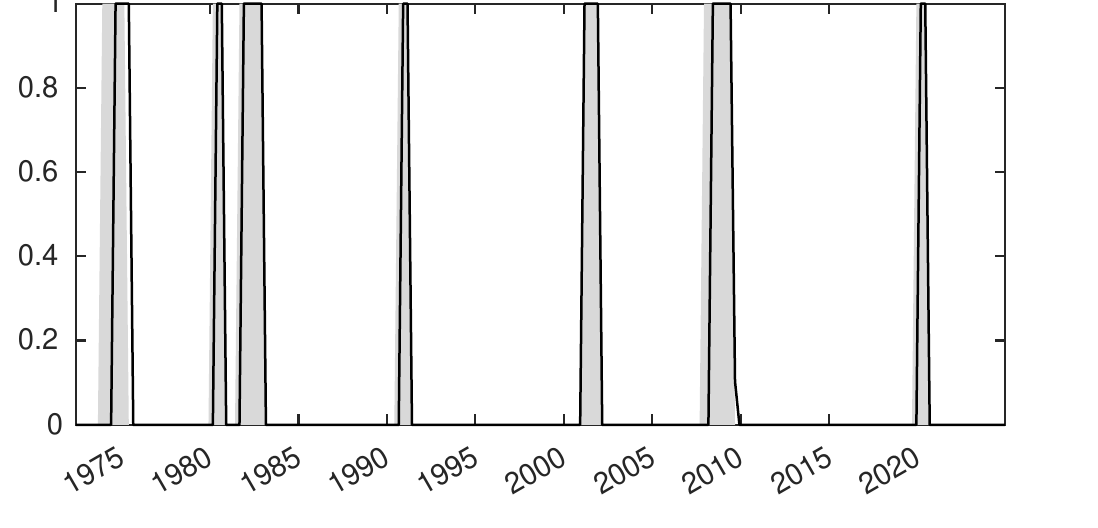}
    \end{minipage}

\vspace{-0.4cm}
\begin{minipage}{\linewidth}
\singlespacing
\footnotesize{ \textit{Note:} The disaggregate series are each modeled as
\(
Y_{i,t} = \delta_{i,1} \mathbf{1}_{\{S_t=1\}} + \mu_{i,2} + e_{i,t},
\)
with $e_{i,t} \sim \text{i.i.d. } N(0,\sigma_{i,1}^2)$ if $S_t=1$ and
$e_{i,t} \sim \text{i.i.d. } N(0,\sigma_{i,2}^2)$ if $S_t=2$, and assumed to be cross-sectionally independent.
The data are quarterly.
Smoothed probabilities are obtained via MCMC, averaging over 0.1 million posterior draws after discarding the first 0.1 million as burn-in.
NBER recession dates are shown as shaded regions.
The prior distributions are described in Section 3.5 and are the same across all applications.}
\end{minipage} 

\end{figure}

At recession onsets, the model dates are typically close to the NBER dates, with four of the six onsets identified within one quarter: the 1980, 1981--82, and 1990--91 recessions are each identified one quarter after the NBER onset; the 2001 recession onset is dated to 2001 Q1 by both. The two exceptions are the 1973--75 recession, where the probability switches to one in 1974 Q3, three quarters after the NBER onset of 1973 Q4, and the 2008--09 recession, where the probability reaches one in 2008 Q2, two quarters after the NBER onset of 2007 Q4. These onset lags may reflect the nature of the underlying data: the early quarters of a downturn are often gradual, and the model achieves high certainty only when a sufficient number of the 94 series have moved in a consistent direction.

The model provides sharp identification of recoveries. In five of the six recessions, the probability of the recession regime drops to (near) zero exactly when the NBER expansion begins. The only exception is the 1975 recovery, where the probability remains at one in 1975 Q2 and drops to zero in 1975 Q3, a one-quarter lag relative to NBER dating. This asymmetry (sometimes moderate delay at onsets, sharp identification of recoveries) is consistent with the view that recessions build gradually across sectors and series, while recoveries tend to be more synchronized and broad-based, which generates a clearer signal in the cross-section of disaggregate data.

We next include the COVID period in the sample. This period is associated with movements in many of the 94 disaggregate series that were extreme by historical standards. This raises the issue of how to handle extreme values. To address this, we winsorize observations in the post-2019 sample that fall outside the 1972–2019 range to the corresponding historical maximum or minimum. That is, if an observation in the extended period exceeds the historical maximum, it is replaced by that maximum; if it falls below the historical minimum, it is replaced by that minimum; otherwise, it is unchanged. The motivation is that for recession detection, extreme values carry little additional information beyond what is already conveyed by a large observation: once a series has moved far enough to be clearly consistent with a recession regime, trimming it has no practical effect on regime classification. At the same time, the adjustment can prevent a small number of pandemic observations from dominating the likelihood, so that parameter estimates would not be pulled away from values that fit the majority of the history. 

The results are shown in Figure~\ref{fig:application_nolag}(b). The results for the 1972-2019 period are essentially the same. For the COVID period, the NBER dates the recession from 2019 Q4\footnote{This is based on the NBER's quarterly dating. The committee noted that the monthly peak (February 2020) occurred in a different quarter (2020 Q1) than the quarterly peak; thus the dating differs across frequencies. See \url{https://www.nber.org/news/business-cycle-dating-committee-announcement-june-8-2020}.} through 2020 Q2. The model dates the recession from 2020 Q1 to 2020 Q2, a one-quarter onset lag in line with most earlier episodes. The post-2020 period through 2024 Q4 is assigned probability zero throughout, including the period of elevated inflation and monetary tightening in 2022-23 which some contemporaneous observers flagged as recession risk.

In summary, the model identifies all NBER-defined recessions in the sample, including the mild recessions of 1990--91 and 2001, the 2008--09 Great Recession, and the COVID recession. Given the simple model specification, this performance likely reflects the informativeness of the disaggregated data. We next turn to richer models that allow for autoregressive components.


\subsection{Models with autoregressive components}
We first consider model (\ref{eq:2.1}) and then (\ref{eq:2.3}), each on the subsample and then the full sample.
\subsubsection{Estimates under Model A}
The smoothed recession probabilities for the 1972--2019 sample are shown in Figure~\ref{fig:Figure_SpecificationA}(a) and can be compared with the baseline Figure~\ref{fig:application_nolag}(a). Again the model detects all NBER-defined recessions with high posterior certainty, and the onset and exit features are similar to the baseline. This reinforces the view that the cross-sectional information from the 94 disaggregate series is the main driver of regime identification.
\begin{figure}[H]
    \centering

    \caption{Smoothed Recession Probability with Autoregressive Dynamics (Model A)}
    \label{fig:Figure_SpecificationA}
    \begin{minipage}[t]{0.5\textwidth}
        \centering
        {\small (a) 1972-2019}\\
        \includegraphics[width=\linewidth, clip, trim=0cm 9cm 0cm 10cm]{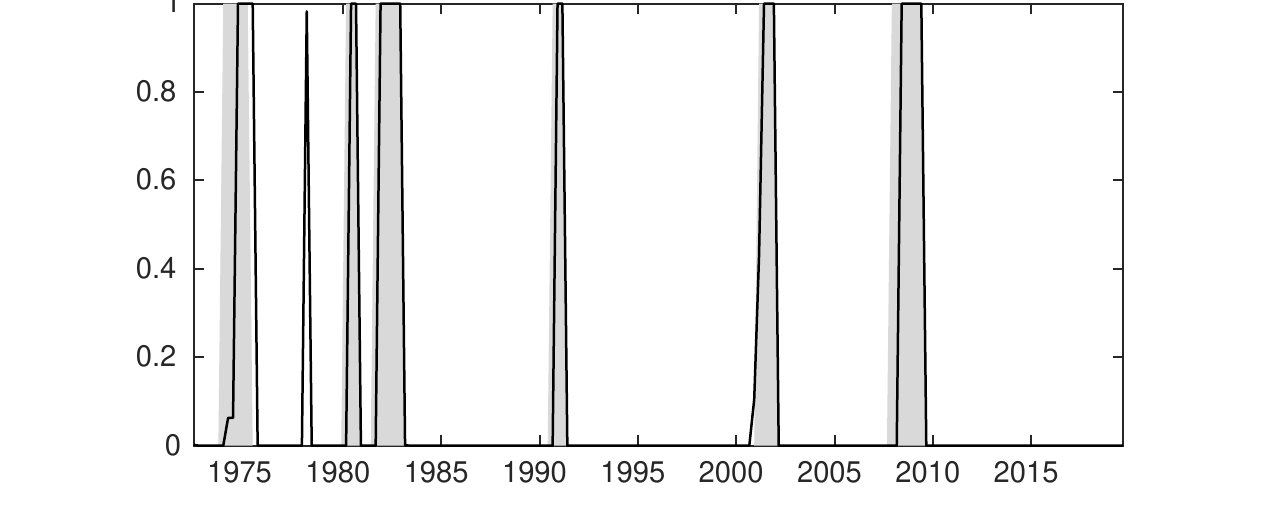}
    \end{minipage}%
    \begin{minipage}[t]{0.5\textwidth}
        \centering
        {\small (b) 1972-2024}\\
        \includegraphics[width=\linewidth, clip, trim=0cm 9cm 0cm 9.5cm]{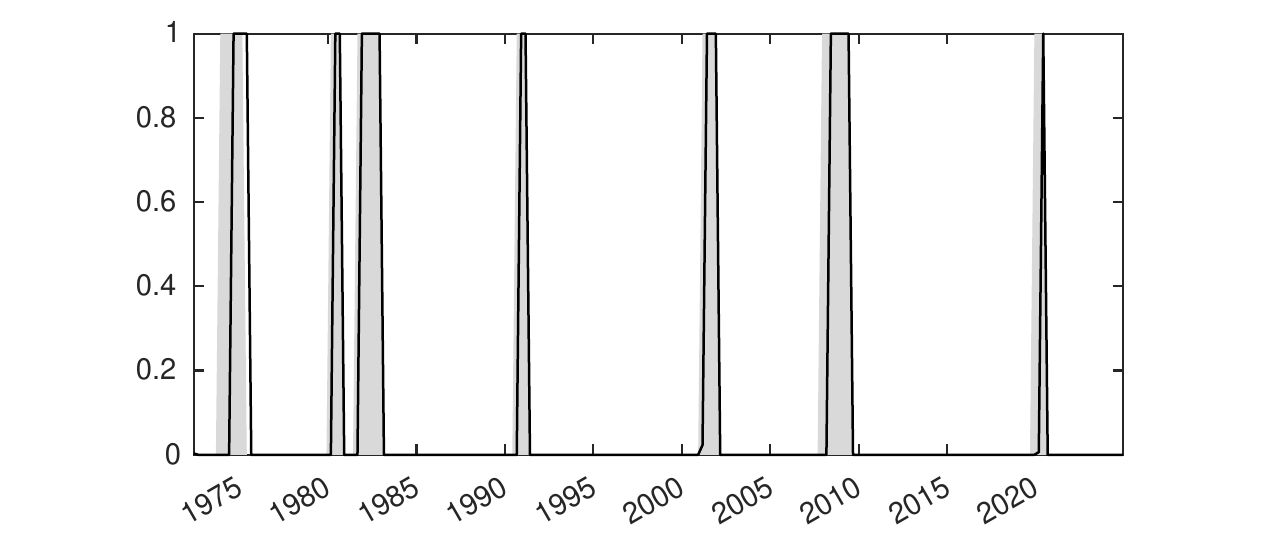}
    \end{minipage}

\vspace{-0.4cm}
\begin{minipage}{\linewidth}
\singlespacing
\footnotesize{ \textit{Note:} The disaggregate series are each modeled as Model A.
Smoothed probabilities are obtained via MCMC, averaging over 0.1 million draws from eight chains after discarding the first 0.1 million as burn-in. NBER recession dates are shown as shaded regions.
The priors are described in Section 3.5, and are the same across all applications.}
\end{minipage} 
\end{figure}
There are only three differences from the baseline case. First, the current model produces a recession probability of 0.982 in 1978 Q1, with no counterpart in the baseline. This likely reflects that, with the autoregressive dynamics absorbing some variation in the data, the model becomes more sensitive to short-term movements in the data. In 1978 Q1, industrial production momentum was somewhat subdued relative to 1975-77, and it was interpreted by the model as a brief recession signal. Second, the 2001 onset is somewhat slower: whereas the baseline model assigns probability one immediately in 2001 Q1, the current model produces only 0.431 that quarter, reaching one in 2001 Q2. Third, the 2008–09 recession exit is sharper: the probability drops to zero in 2009 Q3, with no mass around 0.1 as in the baseline case. Other than these, the onset and exit dates coincide with the baseline.

The results for the full sample are reported in Figure~\ref{fig:Figure_SpecificationA}(b) and can be compared with the baseline results in Figure~\ref{fig:application_nolag}(b). The smoothed probabilities are close for most of the sample, with only a few differences. The 2008-09 recession exit is now sharper: the baseline assigns a residual probability of 0.10 in 2009 Q3, while the current model assigns zero, matching the start of the NBER expansion. The 2001 onset is one quarter slower than the baseline, with probability 0.02 in 2001 Q1, reaching one in 2001 Q2. Similarly, the COVID-period onset is one quarter slower: the baseline reaches one in 2020 Q1 (a one-quarter lag relative to the NBER onset of 2019 Q4), while the current model assigns probability 0.01 in 2020 Q1 and reaches one in 2020 Q2. Overall, adding the AR dynamics leaves the recession dating conclusions largely unchanged.

\subsubsection{Estimates under Model B}

We next consider Model B with lag order one. The results for the 1972-2019 sample are reported in Figure~\ref{fig:Figure_SpecificationB}(a) and can be compared with the baseline and the Model A results discussed above.

The estimates remain close to the baseline, with all NBER-defined recessions identified and broadly similar timing. Compared with the baseline, Model B is more responsive to cyclical signals and the onset lags relative to NBER are shorter in several cases. The 1980 recession is identified in its first NBER quarter (1980 Q1) with probability 0.85, compared with near-zero in the other specifications. For 2008-09, although the recession probability still reaches one in 2008 Q2 (as do the other specifications), the model assigns probability 0.38 in 2008 Q1, closer to the NBER dating.

The trade-off is a tendency to assign elevated probabilities outside NBER recession dates:  0.56 in 1990 Q2 and 0.99 in 1981 Q1 during the short inter-recession expansion. The pattern suggests that Model B trades some posterior certainty for shorter onset lags over this period.


\begin{figure}[H]
    \centering

    \caption{Smoothed Recession Probability with Autoregressive Dynamics (Model B)}
    \label{fig:Figure_SpecificationB}
    \begin{minipage}[t]{0.5\textwidth}
        \centering
        {\small (a) 1972-2019}\\
        \includegraphics[width=\linewidth, clip, trim=0cm 9cm 0cm 10cm]{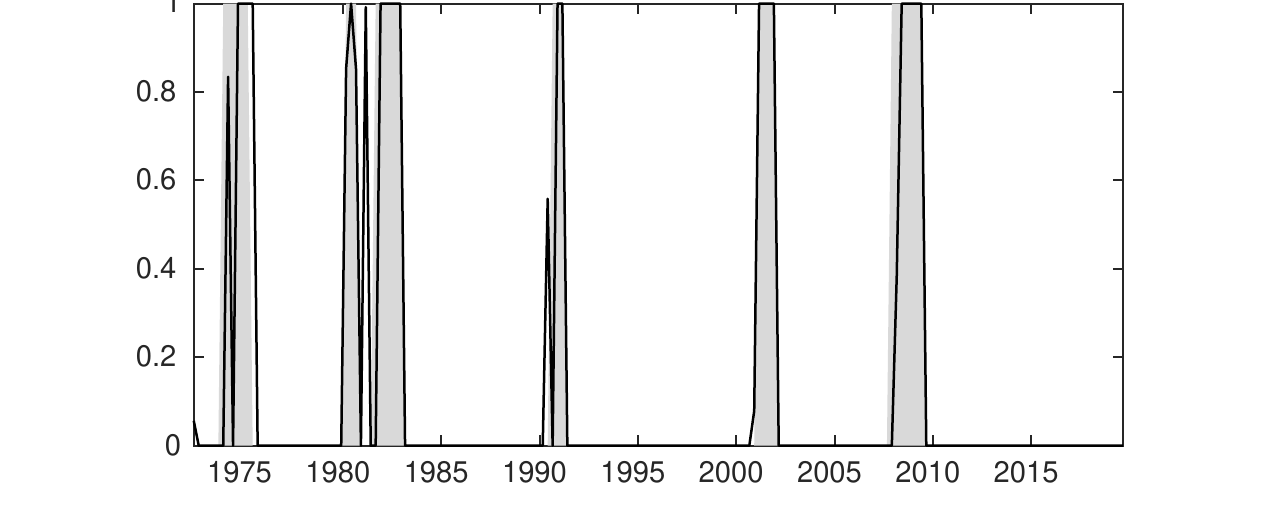}
    \end{minipage}%
    \begin{minipage}[t]{0.5\textwidth}
        \centering
        {\small (b) 1972-2024}\\
        \includegraphics[width=\linewidth, clip, trim=0cm 9cm 0cm 9.5cm]{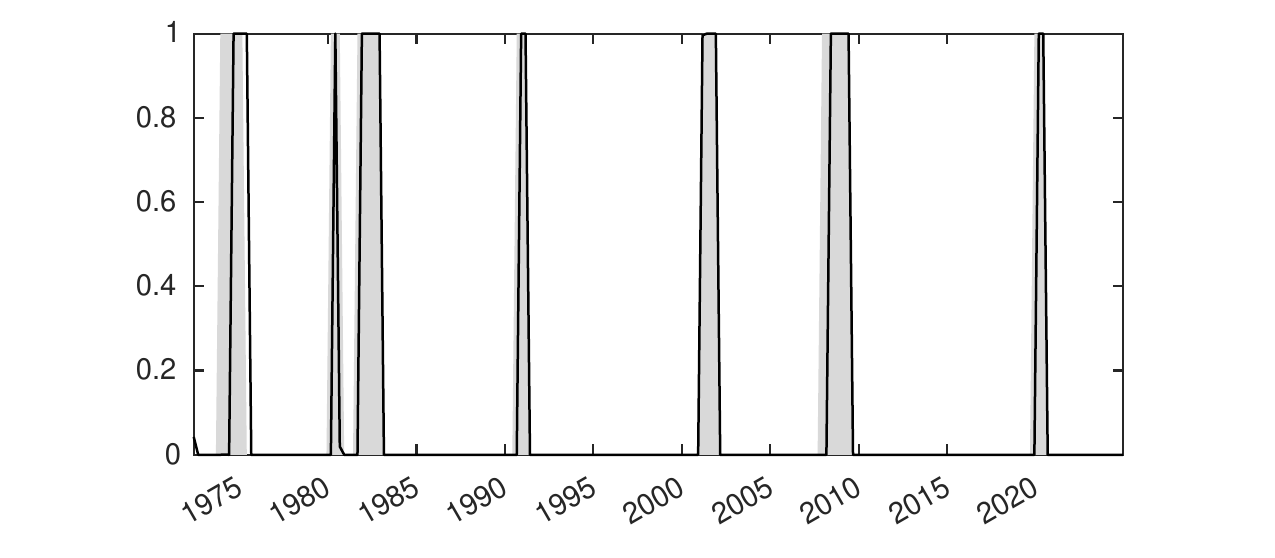}
    \end{minipage}

\vspace{-0.4cm}
\begin{minipage}{\linewidth}
\singlespacing
\footnotesize{ \textit{Note:} The disaggregate series are each modeled as Model B.
Smoothed probabilities are obtained via MCMC, averaging over 0.1 million draws after discarding the first 0.1 million as burn-in.
NBER recession dates are shown as shaded regions.
The priors are described in Section~3.5, and are the same across all applications.}
\end{minipage} 

\end{figure}

The results for the full sample are reported in Figure~\ref{fig:Figure_SpecificationB}(b). The estimates are close to the no-AR baseline (Figure~\ref{fig:application_nolag}(b)), with onset lags and exit dates nearly identical across recessions. The only exception is the 1980 exit, where the probability drops to zero in 1980 Q3, one quarter earlier than the baseline (which drops to zero in 1980 Q4, matching the start of the NBER expansion). Notably, the earlier-onset signals and elevated mid-expansion probabilities seen in the short sample (1974 Q1, 1981 Q1, 1990 Q2, and 2008 Q1) do not appear here. The more responsive behavior of Model B in the short sample appears to be sample-dependent: with the longer sample, the estimates move toward the no-AR baseline in Figure~\ref{fig:application_nolag}(b).

\subsubsection{Summary}

The above results show that across all model specifications and sample periods, the smoothed probabilities are close to zero or one with few intermediate values. This contrasts with typical recession probability estimates, which tend to linger at intermediate values for extended periods.

Models with autoregressive dynamics show higher sensitivity to cyclical signals, which in the short sample produces more variation: for Model A the elevated probability of 0.98 in 1978 Q1, and for Model B elevated probabilities of 0.83 in 1974 Q1, 0.56 in 1990 Q2, and 0.99 in 1981 Q1 during the short inter-recession expansion. Over the full sample, the estimates across specifications are quite consistent, with differences limited to only isolated quarters. The consistency across specifications confirms that the disaggregate data, rather than the model dynamics, drive regime identification.

The model dates either coincide with NBER onsets or lag them by one to three quarters depending on the episode, while exits coincide closely with the NBER expansion start in most cases. This asymmetry may be a reflection of the nature of recessions: downturns build gradually across sectors, while recoveries tend to generate a more synchronized signal across disaggregate series. The degree of agreement achieved is interesting: the model detects all NBER-defined recessions across more than five decades, despite the fact that NBER dating itself involves a significant judgmental component, which no purely statistical model is designed to replicate.

In terms of computational costs (on an Intel Xeon CPU at 2.80 GHz using 8 cores, with 8 chains producing a total of 0.2 million draws), the computational times are: 0.6 hours for the no-lag specification, 3.1 hours for Model A with one lag, and 0.6 hours for Model B with one lag. Model A takes the most time due to the additional state dependence. Together with the simulation cases reported above, these values confirm that the procedure is computationally feasible for empirically relevant model and sample sizes.


\section{Conclusion}

We showed that cross-sectional aggregation attenuates regime-switching signals, and that working with disaggregate data improves state inference under empirically relevant conditions. We then presented models and estimation methods for uncovering regime switching directly from disaggregated data. Applied to a U.S. macroeconomic dataset of 94 series, the models produce nearly binary recession probabilities across all specifications and sample periods, detecting all NBER-defined recessions from 1972 through 2024. The estimated recession onset dates either coincide with or lag the NBER by one to three quarters, while exits closely track the start of the NBER expansion. Applying the models to other datasets (such as European data) or adapting the sampler to related high-dimensional regime-switching models could be useful directions for future research.


\baselineskip=22pt

\bibliography{refs}      

\newpage 

\renewcommand{\thesection}{A.\arabic{section}}
\renewcommand{\thefigure}{A.\arabic{figure}}
\renewcommand{\thetable}{A.\arabic{table}}
\renewcommand{\theequation}{A.\arabic{equation}}
\setcounter{section}{0}
\setcounter{figure}{0}
\setcounter{table}{0}
\setcounter{equation}{0}
\section*{\centering \Huge{Appendix}}
\vspace{0.5cm}
This Appendix consists of four sections: (A.1) proofs of the results in the paper, (A.2) details on the sampler, (A.3) details on the data source and data structure, and (A.4) the figures and diagrams supporting the main paper and other parts of the Appendix.\section{Proofs}
\paragraph{Proof of Proposition 1.}\label{proof1}
We first prove the \textit{if} part of the property and then the \textit{only if} part.

\noindent{\bf{Proof of the ``if'' part:}}
Recall that $y_t=(y_{1,t},y_{2,t})'$ and $Y_t=y_{1,t}+y_{2,t}$. We begin with the initial period $t=1$ and then consider $t>1$ sequentially. As in \citet{hamilton1989}, at $t=1$ the filtered probabilities for the disaggregated model satisfy
\begin{equation*}
\hat{S}_{1|1}^{1}=\frac{f(y_{1}|S_{1}=1)\hat{S}_{1|0}^{1}}{f(y_{1}|S_{1}=1)\hat{S}_{1|0}^{1}+f(y_{1}|S_{1}=2)(1-\hat{S}_{1|0}^{1})}=\frac{\hat{S}_{1|0}^{1}}{\hat{S}_{1|0}^{1}+\frac{f(y_{1}|S_{1}=2)}{f(y_{1}|S_{1}=1)}\left( 1-\hat{S}_{1|0}^{1}\right)};
\end{equation*}
while for the aggregated model, 
\begin{equation*}
\Breve{S}_{1|1}^{1}=\frac{f(Y_{1}|S_{1}=1)\Breve{S}_{1|0}^{1}}{f(Y_{1}|S_{1}=1)\Breve{S}_{1|0}^{1}+f(Y_{1}|S_{1}=2)(1-\Breve{S}_{1|0}^{1})}=\frac{\Breve{S}_{1|0}^{1}}{\Breve{S}_{1|0}^{1}+\frac{f(Y_{1}|S_{1}=2)}{f(Y_{1}|S_{1}=1)}(1-\Breve{S}_{1|0}^{1})}.
\end{equation*}
Let 
\begin{equation*}
D_{1}=\frac{f(y_{1}|S_{1}=2)}{f(y_{1}|S_{1}=1)}\text{  and  }A_{1}=\frac{f(Y_{1}|S_{1}=2)}{f(Y_{1}|S_{1}=1)}
\end{equation*}
denote the likelihood ratios for the two models. Then
\begin{equation*}
\hat{S}_{1|1}^{1}=\frac{\hat{S}_{1|0}^{1}}{\hat{S}_{1|0}^{1}+D_{1}(1-\hat{S}_{1|0}^{1})}\text{  and  }\Breve{S}_{1|1}^{1}=\frac{\Breve{S}_{1|0}^{1}}{\Breve{S}_{1|0}^{1}+A_{1}(1-\Breve{S}_{1|0}^{1})}.
\end{equation*}
Since the two models share the same initial condition $p_1$, we get
\begin{equation*}
\hat{S}_{1|1}^{1}=\frac{p_{1}}{p_{1}+D_{1}(1-p_{1})}\text{ and }\Breve{S}_{1|1}^{1}=\frac{p_{1}}{p_{1}+A_{1}(1-p_{1})}.
\end{equation*}

By the cross-sectional independence, the explicit expressions for $D_{1}$ and $A_{1}$ are 
{
\begin{equation*}
D_{1}=\exp \left( -\frac{(\mu_{1,1}-\mu_{1,2})^{2}}{2\sigma_{1}^{2}}-\frac{(\mu_{2,1}-\mu_{2,2})^{2}}{2\sigma_{2}^{2}}\right) \exp \left( -\frac{(\mu_{1,1}-\mu_{1,2})}{\sigma_{1}^{2}}(y_{1,1}-\mu_{1,1})-\frac{(\mu_{2,1}-\mu_{2,2})}{\sigma_{2}^{2}}(y_{2,1}-\mu_{2,1})\right) .
\end{equation*}}
and
{
\begin{equation*}
A_{1}=\exp \left( -\frac{(\mu_{1,1}-\mu_{1,2}+\mu_{2,1}-\mu_{2,2})^{2}}{2\left( \sigma_{1}^{2}+\sigma_{2}^{2}\right) }\right) \exp \left( - \frac{(\mu_{1,1}-\mu_{1,2}+\mu_{2,1}-\mu_{2,2})}{\sigma_{1}^{2}+\sigma_{2}^{2}}(y_{1,1}+y_{2,1}-\mu_{1,1}-\mu_{2,1})\right).
\end{equation*}}
If $(\mu_{1,1}-\mu_{1,2})/\sigma_{1}^{2}=(\mu_{2,1}-\mu_{2,2})/\sigma_{2}^{2}=\kappa$, then
\begin{equation*}
D_{1}=A_{1}=\exp \left\{ -\frac{\kappa}{2}(\mu_{1,1}-\mu_{1,2}+\mu_{2,1}-\mu_{2,2})\right\} \exp \left\{ -\kappa(y_{1,1}+y_{2,1}-\mu_{1,1}-\mu_{2,1})\right\} .
\end{equation*}
Hence, $\hat{S}_{1|1}^{1}=\Breve{S}_{1|1}^{1}$, and so are
\begin{equation*}
\hat{P}_{1|1}^{1}=\hat{S}_{1|1}^{1}(1-\hat{S}_{1|1}^{1})\text{ and }\Breve{P}_{1|1}^{1}=\Breve{S}_{1|1}^{1}(1-\Breve{S}_{1|1}^{1}).
\end{equation*}

Now consider $t=2$. The prediction step for the disaggregate model yields (recall $p_{11}$ and $p_{21}$ are the transition probabilities of the Markov chain from state 1) 
\begin{eqnarray*}
\hat{S}_{2|1}^{1} &=&P(S_{2}=1|y_{1}), \\
&=&P(S_{2}=1|S_{1}=1,y_{1})P(S_{1}=1|y_{1})+P(S_{2}=1|S_{1}=2,y_{1})P(S_{1}=2|y_{1}), \\
&=&p_{11}P(S_{1}=1|y_{1})+p_{21}P(S_{1}=2|y_{1}), \\
&=&p_{11}\hat{S}_{1|1}^{1}+p_{21}(1-\hat{S}_{1|1}^{1}),
\end{eqnarray*}
and for the aggregate model, 
\begin{equation*}
\Breve{S}_{2|1}^{1}=P(S_{2}=1|Y_{1})=p_{11}\Breve{S}_{1|1}^{1}+p_{21}(1-\Breve{S}_{1|1}^{1}).
\end{equation*}
Since $\hat{S}_{1|1}^{1}=\Breve{S}_{1|1}^{1}$, it follows that $\hat{S}_{2|1}^{1}=\Breve{S}_{2|1}^{1}$. The filtering step at $t=2$ proceeds as at $t=1$: 
\begin{equation*}
\hat{S}_{2|2}^{1}=\frac{f(y_{2}|S_{2}=1)\hat{S}_{2|1}^{1}}{f(y_{2}|S_{2}=1) \hat{S}_{2|1}^{1}+f(y_{2}|S_{2}=2){(1-\hat{S}_{2|1}^{1})}}=\frac{\hat{S}_{2|1}^{1}}{\hat{S}_{2|1}^{1}+\frac{f(y_{2}|S_{2}=2)}{f(y_{2}|S_{2}=1)}(1 - \hat{S}_{2|1}^{1})},
\end{equation*}
and 
\begin{equation*}
\Breve{S}_{2|2}^{1}=\frac{f(Y_{2}|S_{2}=1)\Breve{S}_{2|1}^{1}}{f(Y_{2}|S_{2}=1)\Breve{S}_{2|1}^{1}+f(Y_{2}|S_{2}=2)(1 - \Breve{S}_{2|1}^{1})}=\frac{\Breve{S}_{2|1}^{1}}{\Breve{S}_{2|1}^{1}+\frac{f(Y_{2}|S_{2}=2)}{f(Y_{2}|S_{2}=1)}(1 - \Breve{S}_{2|1}^{1})}.
\end{equation*}
These are equal because $\hat{S}_{2|1}^{1}=\Breve{S}_{2|1}^{1}$ and the likelihood ratios satisfy $f(y_{2}|S_{2}=2)/f(y_{2}|S_{2}=1)=f(Y_{2}|S_{2}=2)/f(Y_{2}|S_{2}=1)$, for the same reason as in the $t=1$ case. Consequently, $\hat{P}_{2|2}^{1}=$ $\Breve{P}_{2|2}^{1}$. The cases with $t>2$ can be proved sequentially using the same argument. 

Now consider smoothed probabilities. We have shown that, for all $t\ge1$,
\[
\hat S^{1}_{t\mid t}=\Breve S^{1}_{t\mid t}
\qquad\text{and}\qquad
\hat S^{1}_{t+1\mid t}=\Breve S^{1}_{t+1\mid t}.
\]
We now show that this implies
\[
\hat S^{1}_{t\mid T}=\Breve S^{1}_{t\mid T}
\qquad\text{for all } t\le T.
\]

For a two-state Markov chain with transition probabilities
$p_{ij}=P(S_{t+1}=j\mid S_t=i)$, the standard backward smoothing recursion is
\begin{equation}
P(S_t=i\mid \mathcal F_T)
=
P(S_t=i\mid \mathcal F_t)
\sum_{j=1}^2
\frac{p_{ij}\,P(S_{t+1}=j\mid \mathcal F_T)}{P(S_{t+1}=j\mid \mathcal F_t)},
\qquad i=1,2,
\label{eq:smoother}
\end{equation}
where $\mathcal F_t$ is the information set for filtering. For the disaggregated model, \eqref{eq:smoother} implies
\begin{equation}
\hat S^{1}_{t\mid T}
=
\hat S^{1}_{t\mid t}
\left(
\frac{p_{11}\,\hat S^{1}_{t+1\mid T}}{\hat S^{1}_{t+1\mid t}}
+
\frac{p_{12}\,\hat S^{2}_{t+1\mid T}}{\hat S^{2}_{t+1\mid t}}
\right),
\label{eq:smoother_hat}
\end{equation}
and, for the aggregated model,
\begin{equation}
\Breve S^{1}_{t\mid T}
=
\Breve S^{1}_{t\mid t}
\left(
\frac{p_{11}\,\Breve S^{1}_{t+1\mid T}}{\Breve S^{1}_{t+1\mid t}}
+
\frac{p_{12}\,\Breve S^{2}_{t+1\mid T}}{\Breve S^{2}_{t+1\mid t}}
\right).
\label{eq:smoother_breve}
\end{equation}

We prove $\hat S^{1}_{t\mid T}=\Breve S^{1}_{t\mid T}$ by backward induction.
At $t=T$, smoothing equals filtering, so
\(
\hat S^{1}_{T\mid T}=\Breve S^{1}_{T\mid T}.
\)
Suppose that for some $t<T$,
\[
\hat S^{j}_{t+1\mid T}=\Breve S^{j}_{t+1\mid T},
\qquad j=1,2.
\]
Using \eqref{eq:smoother_hat}--\eqref{eq:smoother_breve} and the equalities above
\(
\hat S^{j}_{t\mid t}=\Breve S^{j}_{t\mid t}\) and
\(\hat S^{j}_{t+1\mid t}=\Breve S^{j}_{t+1\mid t}\) for $j=1,2$,
we obtain
\[
\hat S^{1}_{t\mid T}
=
\Breve S^{1}_{t\mid T}.
\]
By backward induction, this equality holds for all $t\le T$.

\noindent{\bf{Proof of the ``only if'' part:}}
Assume that the two models share the same initial condition
$\hat S^{1}_{1\mid 0}=\Breve S^{1}_{1\mid 0}=p_1$ with $p_1\in(0,1)$, and that
\[
\hat S^{1}_{t\mid t}=\Breve S^{1}_{t\mid t},\qquad
\hat S^{1}_{t\mid t-1}=\Breve S^{1}_{t\mid t-1},\qquad
\hat S^{1}_{t\mid T}=\Breve S^{1}_{t\mid T}
\quad\text{for all } t\ge1.
\]
We prove this implies
\begin{equation}
\frac{\mu_{1,1}-\mu_{1,2}}{\sigma_1^2}
=
\frac{\mu_{2,1}-\mu_{2,2}}{\sigma_2^2}.
\label{eq:cond_onlyif}
\end{equation}

\medskip\noindent
\noindent{Step 1 (Equality of filtered probabilities implies equality of likelihood ratios).}
At $t=1$, as shown in the ``if" part of the proof, for the the disaggregated model, we have
\[
\hat S^{1}_{1\mid 1}
=
\frac{p_1}{p_1+D_1(1-p_1)},
\qquad
D_1\equiv \frac{f(y_1\mid S_1=2)}{f(y_1\mid S_1=1)},
\]
and for the aggregated model, we have
\[
\Breve S^{1}_{1\mid 1}
=
\frac{p_1}{p_1+A_1(1-p_1)},
\qquad
A_1\equiv \frac{f(Y_1\mid S_1=2)}{f(Y_1\mid S_1=1)}.
\]
By assumption, $\hat S^{1}_{1\mid 1}=\Breve S^{1}_{1\mid 1}$. Since $p_1\in(0,1)$, the map
\[
g(x)\equiv \frac{p_1}{p_1+x(1-p_1)},\qquad x>0,
\]
is strictly decreasing and hence one-to-one. Therefore
\begin{equation}
D_1=A_1.
\label{eq:D1A1}
\end{equation}

\medskip\noindent
\noindent{Step 2 (Compute the log-likelihood ratios).}
Write $\delta_1\equiv \mu_{1,1}-\mu_{1,2}$ and $\delta_2\equiv \mu_{2,1}-\mu_{2,2}$.
A direct calculation yields
\begin{equation}
\log D_1
=
-\frac{\delta_1}{\sigma_1^2}\,y_{1,1}
-\frac{\delta_2}{\sigma_2^2}\,y_{2,1}
+
\frac{\delta_1}{2\sigma_1^2}\,(\mu_{1,1}+\mu_{1,2})
+
\frac{\delta_2}{2\sigma_2^2}\,(\mu_{2,1}+\mu_{2,2}).
\label{eq:logD1}
\end{equation}
Equivalently, $\log D_1$ is affine in $(y_{1,1},y_{2,1})$ with slopes
$\delta_1/\sigma_1^2$ and $\delta_2/\sigma_2^2$.

For the aggregated model, the same calculation gives
\begin{equation}
\log A_1
=
-\frac{\delta_1+\delta_2}{\sigma^2}\,(y_{1,1}+y_{2,2})
+
\frac{\delta}{2\sigma^2}\,
\bigl[(\mu_{1,1}+\mu_{2,1})+(\mu_{1,2}+\mu_{2,2})\bigr].
\label{eq:logA1}
\end{equation}
Thus $\log A_1$ is affine in $(y_{1,1},y_{2,1})$ with \emph{the same slope}
on $y_{1,1}$ and $y_{2,1}$, namely $(\delta_1+\delta_2)/(\sigma_1^2+\sigma_2^2)$.

\medskip\noindent
\noindent{Step 3 (Equality of ratios for all realizations implies equality of slopes).}
By \eqref{eq:D1A1}, we have $\log D_1=\log A_1$.
Both sides are affine functions of $(y_{1,1},y_{2,1})$.
Equality of two affine functions on $\mathbb R^2$ implies equality of their coefficients.
Comparing the coefficients on $y_{1,1}$ and $y_{2,1}$ yields
\begin{equation}
\frac{\delta_1}{\sigma_1^2}
=
\frac{\delta_1+\delta_2}{\sigma_1^2+\sigma_2^2},
\qquad
\frac{\delta_2}{\sigma_2^2}
=
\frac{\delta_1+\delta_2}{\sigma_1^2+\sigma_2^2}.
\label{eq:compare_slopes}
\end{equation}
Either identity in \eqref{eq:compare_slopes} implies the desired condition. For example,
from the first equality,
\[
\frac{\delta_1}{\sigma_1^2}(\sigma_1^2+\sigma_2^2)=\delta_1+\delta_2
\quad\Longrightarrow\quad
\delta_1\frac{\sigma_2^2}{\sigma_1^2}=\delta_2
\quad\Longrightarrow\quad
\frac{\delta_1}{\sigma_1^2}=\frac{\delta_2}{\sigma_2^2},
\]
which is \eqref{eq:cond_onlyif}.
This completes the proof of the ``only if'' direction of Proposition 1.\ $\blacksquare $

\noindent{\textbf{Proof of Corollary 1}.} Let $I_t=\mathbf 1_{\{S_t=1\}}$ and define disaggregate and aggregate information sets
\[
\mathcal F_t^d=\sigma(y_1,\ldots,y_t),
\qquad
\mathcal F_t^a=\sigma(Y_1,\ldots,Y_t),
\qquad
Y_t=y_{1,t}+y_{2,t}.
\]
Since both models are without misspecification, we have
\[
\hat S^1_{t\mid t}=E(I_t\mid \mathcal F_t^d),
\qquad
\Breve S^1_{t\mid t}=E(I_t\mid \mathcal F_t^a).
\]
Since $Y_s$ is a measurable function of $y_s$, we have
$\mathcal F_t^a\subseteq \mathcal F_t^d$. By the law of total variance,
\[
\operatorname{Var}(I_t\mid\mathcal F_t^a)
=
E\!\left[\operatorname{Var}(I_t\mid\mathcal F_t^d)\mid\mathcal F_t^a\right]
+
\operatorname{Var}\!\left(E(I_t\mid\mathcal F_t^d)\mid\mathcal F_t^a\right).
\]
Taking expectations yields
\begin{equation}
E\!\left[\operatorname{Var}(I_t\mid\mathcal F_t^a)\right]
=
E\!\left[\operatorname{Var}(I_t\mid\mathcal F_t^d)\right]
+
E\!\left[\operatorname{Var}\!\left(\hat S^1_{t\mid t}\mid\mathcal F_t^a\right)\right]
\;\ge\;
E\!\left[\operatorname{Var}(I_t\mid\mathcal F_t^d)\right].
\label{eq:total_variance_decomp}
\end{equation}
Since $I_t$ is Bernoulli given either $\sigma$-field,
\(
\operatorname{Var}(I_t|\mathcal F)
=
E(I_t|\mathcal F)(1-E(I_t|\mathcal F)),
\)
implying \eqref{eq:aggregation_loss}. 

Because $\hat S^1_{t\mid t}$ is nonnegative, equality in (\ref{eq:total_variance_decomp}) holds if and only if
\[
\operatorname{Var}\!\left(\hat S^1_{t\mid t}\mid\mathcal F_t^a\right)=0
\quad \text{a.s.},
\]
that is, if and only if $\hat S^1_{t\mid t}$ is $\mathcal F_t^a$-measurable. Since
\[
E(\hat S^1_{t\mid t}\mid\mathcal F_t^a)
=
E(I_t\mid\mathcal F_t^a)
=
\Breve S^1_{t\mid t},
\]
this condition is equivalent to
$\hat S^1_{t\mid t}=\Breve S^1_{t\mid t}$ almost surely.

Proposition 1 shows that
$\hat S^1_{t\mid t}=\Breve S^1_{t\mid t}$ for all $t\ge1$ if and only if
 \eqref{cond-1} holds. If  \eqref{cond-1} fails,
then already at $t=1$ the filtered probabilities differ with positive probability,
so
\[
\operatorname{Var}\!\left(\hat S^1_{1\mid 1}\mid\mathcal F_1^a\right)>0,
\]
and the inequality in \eqref{eq:aggregation_loss} is strict. 
This completes the proof Corollary 1.\ $\blacksquare $

\noindent{\bf{Proof of Proposition 2.}}\label{proof2} 
As in the switching in means case, we write 
\begin{equation*}
\hat{S}_{1|1}^{1}=\frac{p_{1}}{p_{1}+D_{1}(1-p_{1})}\text{ and }\Breve{S}_{1|1}^{1}=\frac{p_{1}}{p_{1}+A_{1}(1-p_{1})}.
\end{equation*}
The explicit expressions for $D_{1}$ and $A_{1}$ are now
\begin{eqnarray*}
D_{1} &=&\frac{f(y_{1}|S_{1}=2)}{f(y_{1}|S_{1}=1)}=\frac{\frac{1}{\sigma_{1,2}}\exp \left( -\frac{(y_{1,1}-\mu_{1})^{2}}{2\sigma_{1,2}^{2}}\right) \frac{1}{\sigma_{2,2}}\exp \left( -\frac{(y_{2,1}-\mu _{2})^{2}}{2\sigma_{2,2}^{2}}\right) }{\frac{1}{\sigma_{1,1}}\exp \left( -\frac{(y_{1,1}-\mu_{1})^{2}}{2\sigma_{1,1}^{2}}\right) \frac{1}{\sigma_{2,1}}\exp \left( -\frac{(y_{2,1}-\mu_{2})^{2}}{2\sigma_{2,1}^{2}}\right) }, \\
A_{1} &=&\frac{f(Y_{1}|S_{1}=2)}{f(Y_{1}|S_{1}=1)}=\frac{\frac{1}{\sqrt{\sigma_{1,2}^{2}+\sigma_{2,2}^{2}}}\exp \left( -\frac{(y_{1,1}+y_{2,1}-\mu_{1}-\mu_{2})^{2}}{2\left( \sigma _{1,2}^{2}+\sigma_{2,2}^{2}\right) }\right) }{\frac{1}{\sqrt{\sigma _{1,1}^{2}+\sigma_{2,1}^{2}}}\exp \left( - \frac{(y_{1,1}+y_{2,1}-\mu_{1}-\mu_{2})^{2}}{2\left( \sigma_{1,1}^{2}+\sigma_{2,1}^{2}\right) }\right) }.
\end{eqnarray*}
To simplify the notation, we define $\tilde{y}_{1,1}=y_{1,1}-\mu _{1}$ and $\tilde{y}_{2,1}=y_{2,1}-\mu _{2}$. Then,
\begin{eqnarray*}
D_{1} &=&\frac{\sigma _{1,1}\sigma _{2,1}}{\sigma _{1,2}\sigma _{2,2}}\exp \left( -\frac{\tilde{y}_{1,1}^{2}}{2\sigma _{1,2}^{2}}+\frac{\tilde{y}_{1,1}^{2}}{2\sigma _{1,1}^{2}}-\frac{\tilde{y}_{2,1}^{2}}{2\sigma _{2,2}^{2}}+\frac{\tilde{y}_{2,1}^{2}}{2\sigma _{2,1}^{2}}\right) , \\
A_{1} &=&\frac{\sqrt{\sigma _{1,1}^{2}+\sigma _{2,1}^{2}}}{\sqrt{\sigma_{1,2}^{2}+\sigma _{2,2}^{2}}}\exp \left( -\frac{(\tilde{y}_{1,1}+\tilde{y}_{2,1})^{2}}{2\left( \sigma _{1,2}^{2}+\sigma _{2,2}^{2}\right) }+\frac{(\tilde{y}_{1,1}+\tilde{y}_{2,1})^{2}}{2\left( \sigma_{1,1}^{2}+\sigma_{2,1}^{2}\right) }\right) .
\end{eqnarray*}
This implies
\begin{eqnarray*}
D_{1} &=&\frac{\sigma_{1,1}\sigma_{2,1}}{\sigma_{1,2}\sigma_{2,2}}\exp \left( -\frac{\tilde{y}_{1,1}^{2}}{2\sigma_{1,2}^{2}}+\frac{\tilde{y}_{1,1}^{2}}{2\sigma_{1,1}^{2}}-\frac{\tilde{y}_{2,1}^{2}}{2\sigma_{2,2}^{2}}+\frac{\tilde{y}_{2,1}^{2}}{2\sigma_{2,1}^{2}}\right) , \\
A_{1} &=&\frac{\sqrt{\sigma_{1,1}^{2}+\sigma_{2,1}^{2}}}{\sqrt{\sigma_{1,2}^{2}+\sigma _{2,2}^{2}}} \\
&\times &\exp \left( -\frac{\tilde{y}_{1,1}^{2}+\tilde{y}_{2,1}^{2}}{2\left( \sigma _{1,2}^{2}+\sigma_{2,2}^{2}\right) }+\frac{\tilde{y}_{1,1}^{2}+\tilde{y}_{2,1}^{2}}{2\left( \sigma_{1,1}^{2}+\sigma_{2,1}^{2}\right) }-\frac{\tilde{y}_{1,1}\tilde{y}_{2,1}}{\sigma_{1,2}^{2}+\sigma_{2,2}^{2}}+\frac{\tilde{y}_{1,1}\tilde{y}_{2,1}}{
\sigma_{1,1}^{2}+\sigma _{2,1}^{2}}\right).
\end{eqnarray*}
More explicitly,
\begin{eqnarray*}
D_{1} &=&\frac{\sigma_{1,1}\sigma_{2,1}}{\sigma_{1,2}\sigma _{2,2}}\exp \left( \tilde{y}_{1,1}^{2}\left(\frac{1}{2\sigma_{1,1}^{2}}-\frac{1}{2\sigma_{1,2}^{2}}\right) +\tilde{y}_{2,1}^{2}\left( \frac{1}{2\sigma_{2,1}^{2}}-\frac{1}{2\sigma_{2,2}^{2}}\right) \right) , \\
A_{1} &=&\frac{\sqrt{\sigma_{1,1}^{2}+\sigma_{2,1}^{2}}}{\sqrt{\sigma_{1,2}^{2}+\sigma_{2,2}^{2}}}\exp \left\{ \tilde{y}_{1,1}^{2}\left( \frac{1}{2\left( \sigma_{1,1}^{2}+\sigma_{2,1}^{2}\right) }-\frac{1}{2\left( \sigma_{1,2}^{2}+\sigma_{2,2}^{2}\right) }\right) \right. \\
&&\left. +\tilde{y}_{2,1}^{2}\left( \frac{1}{2\left( \sigma_{1,1}^{2}+\sigma_{2,1}^{2}\right) }-\frac{1}{2\left( \sigma_{1,2}^{2}+\sigma_{2,2}^{2}\right) }\right)\right. \\
&&\left. +\tilde{y}_{1,1}\tilde{y}_{2,1}\left( \frac{1}{\sigma_{1,1}^{2}+\sigma_{2,1}^{2}}-\frac{1}{\sigma_{1,2}^{2}+\sigma_{2,2}^{2}}\right) \right\} .
\end{eqnarray*}
Because $(\tilde{y}_{1,1},\tilde{y}_{2,1})$ has a non-degenerate continuous
distribution, the equality $D_{1}=A_{1}$ almost surely requires the coefficient
on $\tilde{y}_{1,1}\tilde{y}_{2,1}$ in $\log A_{1}$ to be zero. Hence
\[
\sigma_{1,1}^{2}+\sigma_{2,1}^{2}
=
\sigma_{1,2}^{2}+\sigma_{2,2}^{2}.
\] Under this condition, the coefficients of $\tilde{y}_{1,1}^{2}$ and
$\tilde{y}_{2,1}^{2}$ in $A_{1}$ also vanish. For $D_{1}=A_{1}$ to hold,
the corresponding coefficients in $D_{1}$ must be zero as well, which implies
$\sigma_{1,1}=\sigma_{1,2}$ and $\sigma_{2,1}=\sigma_{2,2}$. Therefore, there is
no regime switching in the variances of the data. This proves the result. $\blacksquare$

\noindent{\bf{Proof of Corollary 2}}.
Let $I_t=\mathbf{1}_{\{S_t=1\}}$ and define the information sets
\[
\mathcal{F}_t^{(d)}=\sigma\!\bigl((y_{1,s},y_{2,s})_{1\le s\le t}\bigr),
\qquad
\mathcal{F}_t^{(a)}=\sigma\!\bigl((Y_s)_{1\le s\le t}\bigr),
\]
so that $\mathcal{F}_t^{(a)}\subseteq \mathcal{F}_t^{(d)}$ and
$\hat S^1_{t\mid t}=E(I_t\mid \mathcal{F}_t^{(d)})$, $\Breve S^1_{t\mid t}=E(I_t\mid \mathcal{F}_t^{(a)})$.
Since $I_t$ is Bernoulli conditional on either information set,
\[
Var(I_t\mid \mathcal{F}_t^{(d)})=\hat S^1_{t\mid t}\bigl(1-\hat S^1_{t\mid t}\bigr),
\qquad
Var(I_t\mid \mathcal{F}_t^{(a)})=\Breve S^1_{t\mid t}\bigl(1-\Breve S^1_{t\mid t}\bigr).
\]
By the law of total variance with $\mathcal{F}_t^{(a)}\subseteq \mathcal{F}_t^{(d)}$,
\[
Var(I_t\mid \mathcal{F}_t^{(a)})
=
E\!\bigl[Var(I_t\mid \mathcal{F}_t^{(d)})\mid \mathcal{F}_t^{(a)}\bigr]
+
Var\!\bigl(E(I_t\mid \mathcal{F}_t^{(d)})\mid \mathcal{F}_t^{(a)}\bigr),
\]
and taking expectations yields \eqref{eq:aggregation_loss_var_weak}. Because second term on the
right-hand side is nonnegative, the reversed inequality is impossible for any $t$.

If $\sigma_{1,1}^{2}\neq\sigma_{1,2}^{2}$ or $\sigma_{2,1}^{2}\neq\sigma_{2,2}^{2}$, then
Proposition~2 implies $\hat S^1_{1\mid 1}\neq \Breve S^1_{1\mid 1}$ with probability one.
Hence $\hat S^1_{1\mid 1}$ cannot be a function of $Y_1$ alone, which implies
$Var(\hat S^1_{1\mid 1}\mid \mathcal{F}_1^{(a)})>0$ with positive probability. Thus the
inequality is strict at $t=1$. This proves the result. $\blacksquare$
\section{Additional Details on the Sampler}
This section follows the same structure as Section 3 of the main paper. Without loss of generality, we assume all series belong to one group, so that $J=1$. This allows us to drop the $j$ index from all expressions.
\vspace{-0.3cm}
\subsection{Sampling regression coefficients and variances}
\subsubsection{Sampling state-dependent means}
Here all distributions are understood to be conditional on
$\{\phi_m\}_{m=1}^k$, $\{\sigma_{i,s}^2\}$, and $S_{1:T}$. 
First, we vectorize the observation equations. Let
\[
\tilde y_{i,t}
=
y_{i,t}-\sum_{m=1}^{k}\phi_{m}y_{i,t-m},
\qquad
z_t
=
\mathbf{1}_{\{S_{t}=1\}}
-
\sum_{m=1}^{k}\phi_{m}\mathbf{1}_{\{S_{t-m}=1\}},
\qquad
c
=
1-\sum_{m=1}^{k}\phi_m .
\]
Then the model can be written as
\[
\tilde y_{i,t}
=
\delta_{i,1} z_t + \mu_{i,2} c + e_{i,t}.
\]
Define
\begin{equation*}
Y_{(NT\times 1)}=%
\begin{pmatrix}
\tilde y_{1} \\ 
\tilde y_{2} \\ 
\vdots \\ 
\tilde y_{N}%
\end{pmatrix}%
,\quad
X_{(NT\times 2N)}=I_{(N\times N)}\otimes 
\begin{pmatrix}
z_{1} & c \\ 
z_{2} & c \\ 
\vdots & \vdots \\ 
z_{T} & c%
\end{pmatrix}%
,\quad
\mu_{(2N\times 1)}=%
\begin{pmatrix}
\delta_{1,1} \\ 
\mu_{1,2} \\ 
\vdots \\ 
\delta_{N,1} \\ 
\mu_{N,2}%
\end{pmatrix}%
,\quad
U_{(NT\times 1)}=%
\begin{pmatrix}
e_{1} \\ 
e_{2} \\ 
\vdots \\ 
e_{N}%
\end{pmatrix}%
,
\end{equation*}%
where, for each \(i\), \(\tilde y_i=(\tilde y_{i,1},\ldots,\tilde y_{i,T})'\) and \(e_i=(e_{i,1},\ldots,e_{i,T})'\).
Then,
\[
Y=X\mu+U.
\]
Because $e_{i,t}\sim N(0,\sigma _{i,S_{t}}^{2})$, the variance covariance
matrix of the model is
\begin{equation*}
E(UU^{\prime })=V_{1}\otimes \text{diag}\{\mathbf{1}_{\{S_{1}=1\}},...,%
\mathbf{1}_{\{S_{T}=1\}}\}+V_{2}\otimes \text{diag}\{\mathbf{1}%
_{\{S_{1}=2\}},...,\mathbf{1}_{\{S_{T}=2\}}\}\equiv V_{(NT\times NT)},
\end{equation*}%
where 
\begin{equation*}
V_{1}=%
\begin{pmatrix}
\sigma _{1,1}^{2} & 0 & \cdots & 0 \\ 
0 & \sigma _{2,1}^{2} & \cdots & 0 \\ 
\vdots & \vdots & \ddots & \vdots \\ 
0 & 0 & \cdots & \sigma _{N,1}^{2}%
\end{pmatrix}%
,\text{ and }V_{2}=%
\begin{pmatrix}
\sigma _{1,2}^{2} & 0 & \cdots & 0 \\ 
0 & \sigma _{2,2}^{2} & \cdots & 0 \\ 
\vdots & \vdots & \ddots & \vdots \\ 
0 & 0 & \cdots & \sigma _{N,2}^{2}%
\end{pmatrix}%
.
\end{equation*}%
Next, pre-multiply the vectorized model by $V^{-1/2}$. Then,
we have the following transformed observation equation 
\begin{equation*}
Y^{\ast }=X^{\ast }\mu +U^{\ast },
\end{equation*}%
where 
\begin{equation*}
Y^{\ast }=V^{-1/2}Y,X^{\ast }=V^{-1/2}X,U^{\ast }=V^{-1/2}U.
\end{equation*}%
and 
\begin{equation*}
E(U^{\ast }U^{^{\prime }\ast })=I_{(NT\times NT)}.
\end{equation*}%

If the prior distribution of $\mu$ is specified as a multivariate normal with
mean $\mu^{\ast}$ and covariance matrix $V^{\ast}$, subject to the sign
restrictions $\delta_{i,1}\le 0$ for $i=1,\ldots,N$, that is,
\[
\mu \sim \mathcal{TN}_{(-\infty,0]^N}\bigl(\mu^{\ast},V^{\ast}\bigr),
\]
the posterior distribution is
\[
\mu \mid y_{1:T}
\sim \mathcal{TN}_{(-\infty,0]^N}\bigl(\mu_{\ast},V_{\ast}\bigr),
\]
where
\begin{align*}
V_{\ast}
&=
\bigl[V^{\ast -1}+X^{\ast\prime}X^{\ast}\bigr]^{-1},
\\
\mu_{\ast}
&=
V_{\ast}\bigl[V^{\ast -1}\mu^{\ast}+X^{\ast\prime}Y^{\ast}\bigr].
\end{align*}
Posterior simulation can be carried out by drawing from the truncated multivariate
normal distribution using a minimax-tilting accept–reject sampler, which
efficiently enforces the linear inequality constraints.
\subsubsection{Sampling the variances}
Here all distributions are conditional on the latent states and the regression coefficients.
\subsubsection*{Case 1. The variances are not allowed to switch}

Since we assume the series are conditionally independent, and variance covariance matrix
is diagonal, we can draw $\sigma _{i}^{2}$ one series at a time. As common in this literature, we use independent inverse Gamma priors: 
\begin{equation*}
\sigma _{i}^{2}\sim IG(d_{i,0}/2,v_{i,0}/2),i=1,2,...,N.
\end{equation*}%
The posterior distribution is then given by 
\begin{equation*}
\sigma _{i}^{2}|{y}_{1:T}\sim IG(d_{i,1}/2,v_{i,1}/2),
\end{equation*}%
where 
\begin{equation*}
d_{i,1}=d_{i,0}+T,\text{ and }v_{i,1}=v_{i,0}+SSE_{i},
\end{equation*}%
where $SSE_{i}$ is the the sum of
squared residuals for time series $i$:
\begin{equation*}
SSE_i
=
\sum_{t=1}^{T}
\left[
y_{i,t}
-
\mu_{i,S_t}
-
\sum_{m=1}^{k}\phi_m\bigl(y_{i,t-m}-\mu_{i,S_{t-m}}\bigr)
\right]^2 .
\end{equation*}
For each $i$, the sampling is a one dimensional problem not involving any matrix inversion.

\subsubsection*{Case 2. The variances are allowed to switch}
For each $i$, the sampling is identical to that in \citet{kim1999}. Note that for Model~A,
\[e_{i,t}=y_{i,t}-
\mu_{i,S_t}
-
\sum_{m=1}^{k}\phi_m\bigl(y_{i,t-m}-\mu_{i,S_{t-m}}\bigr),
\]
with conditional distribution
\[
e_{i,t}
\sim
i.i.d.\ N\!\left(0,\,
\sigma_{i,1}^{2}\mathbf{1}_{\{S_t=1\}}
+
\sigma_{i,2}^{2}\mathbf{1}_{\{S_t=2\}}
\right).
\]
The regime-dependent variances can be written as
\begin{equation*}
\sigma_{i,S_t}^{2}
=
\sigma_{i,2}^{2}
\left[
1+
\left(
\frac{\sigma_{i,1}^{2}}{\sigma_{i,2}^{2}}-1
\right)\mathbf{1}_{\{S_t=1\}}
\right]
=
\sigma_{i,2}^{2}\bigl[1+h_{i,2}\mathbf{1}_{\{S_t=1\}}\bigr],
\end{equation*}
where $h_{i,2}=\sigma_{i,1}^{2}/\sigma_{i,2}^{2}-1$. We can sample $\sigma_{i,2}^{2}$ given $h_{i,2}$, and then $h_{i,2}+1$, equivalently $\sigma_{i,1}^{2}$,  given  $\sigma_{i,2}^{2}$ as in \citet{kim1999}.

\paragraph{\textit{(i) Generate }$\sigma_{i,2}^{2}$.}

To generate $\sigma_{i,2}^{2}$, we condition on 
$h_{i,2}\mathbf{1}_{\{S_t=1\}}$ and work with the transformed regression of Model A:
\begin{equation*}
\frac{y_{i,t}}{\sqrt{1+h_{i,2}\mathbf{1}_{\{S_t=1\}}}}
=
\frac{
\mu_{i,S_t}
+
\sum_{m=1}^{k}\phi_m\bigl(y_{i,t-m}-\mu_{i,S_{t-m}}\bigr)
}
{\sqrt{1+h_{i,2}\mathbf{1}_{\{S_t=1\}}}}
+
\frac{e_{i,t}}{\sqrt{1+h_{i,2}\mathbf{1}_{\{S_t=1\}}}} .
\end{equation*}
In this transformed regression, the error term
$e_{i,t}/\sqrt{1+h_{i,2}\mathbf{1}_{\{S_t=1\}}}$ is i.i.d.\ $N(0,\sigma_{i,2}^{2})$.

With the prior distribution
\begin{equation*}
\sigma_{i,2}^{2}\sim \mathrm{IG}\!\left(\frac{d_{i,0}}{2},\frac{v_{i,0}}{2}\right).
\end{equation*}
the posterior distribution is
\begin{equation*}
\sigma_{i,2}^{2}\mid h_{i,2},y_{1:T}
\sim
\mathrm{IG}\!\left(\frac{d_{i,2}}{2},\frac{v_{i,2}}{2}\right),
\end{equation*}
where
\begin{equation*}
d_{i,2}=d_{i,0}+T,
\qquad
v_{i,2}=v_{i,0}+SSE_i,
\end{equation*}
$T$ is the sample size, and $SSE_i$ is the sum of squared transformed residuals for series
$i$ over the full sample,
\begin{equation*}
SSE_i
=
\sum_{t=1}^{T}
\left(
\frac{
y_{i,t}
-
\mu_{i,S_t}
-
\sum_{m=1}^{k}\phi_m\bigl(y_{i,t-m}-\mu_{i,S_{t-m}}\bigr)
}
{\sqrt{1+h_{i,2}\mathbf{1}_{\{S_t=1\}}}}
\right)^2 .
\end{equation*}

\paragraph{\textit{(ii) Generate }$\sigma_{i,1}^{2}$.}

To generate $\sigma_{i,1}^{2}$, we condition on $\sigma_{i,2}^{2}$ and again
work with a transformation of Model~A:
\begin{equation*}
\frac{y_{i,t}}{\sigma_{i,2}}
=
\frac{
\mu_{i,S_t}
+
\sum_{m=1}^{k}\phi_m\bigl(y_{i,t-m}-\mu_{i,S_{t-m}}\bigr)
}{\sigma_{i,2}}
+
\frac{e_{i,t}}{\sigma_{i,2}} .
\end{equation*}
In this transformed regression, the error term satisfies
\[
e_{i,t}/\sigma_{i,2}\mid \sigma_{i,2}
\sim
i.i.d.\ N\!\left(0,\,1+h_{i,2}\mathbf{1}_{\{S_t=1\}}\right),
\] which shows that only the observations with $S_t=1$ contain information about $h_{i,2}$

Let $\bar h_{i,2}=1+h_{i,2}$. With the prior
distribution
\begin{equation*}
\bar h_{i,2}\sim \mathrm{IG}\!\left(\frac{d_{i,0}}{2},\frac{v_{i,0}}{2}\right).
\end{equation*}
its posterior distribution is
\begin{equation*}
\bar h_{i,2}\mid \sigma_{i,2}, y_{1:T}
\sim
\mathrm{IG}\!\left(\frac{d_{i,1}}{2},\frac{v_{i,1}}{2}\right),
\end{equation*}
where
\begin{equation*}
d_{i,1}=d_{i,0}+\sum_{t=1}^{T}\mathbf{1}_{\{S_t=1\}},
\qquad
v_{i,1}=v_{i,0}+SSE_{i,S_t=1},
\end{equation*}
with
\[
SSE_{i,S_t=1}
=
\sum_{t:S_t=1}
\left(
\frac{
y_{i,t}
-
\mu_{i,S_t}
-
\sum_{m=1}^{k}\phi_m\bigl(y_{i,t-m}-\mu_{i,S_{t-m}}\bigr)
}{\sigma_{i,2}}
\right)^2 .
\]
Finally, $\sigma_{i,1}^{2}$ is obtained as
\[
\sigma_{i,1}^{2}=\bar h_{i,2}\,\sigma_{i,2}^{2}.
\]
\subsubsection{Sampling the autoregressive coefficients}

Dividing both sides of Model A by $\sigma_{i,S_t}$ yields the transformed regression
\[
\frac{y_{i,t}-\mu_{i,S_t}}{\sigma_{i,S_t}}
=
\sum_{m=1}^{k}\phi_m
\frac{y_{i,t-m}-\mu_{i,S_{t-m}}}{\sigma_{i,S_t}}
+
u_{i,t},
\qquad
u_{i,t}\sim i.i.d.\ N(0,1).
\]
The likelihood is homoskedastic in the transformed system. We stack the transformed regression over $t=k+1,\ldots,T$ and over all series
in a given group. Let $Y$ denote the stacked vector with
entries $(y_{i,t}-\mu_{i,S_t})/\sigma_{i,S_t}$, and let $X$ be the
corresponding stacked design matrix whose $m$th column contains
$(y_{i,t-m}-\mu_{i,S_{t-m}})/\sigma_{i,S_t}$. This yields
\[
Y = X\phi + u,
\qquad
u\sim N(0,I).
\]

With a Gaussian prior $\phi\sim N(m_0,V_0)$, the conditional posterior
distribution is Gaussian,
\[
\phi\mid Y \sim N(m_1,V_1),
\qquad
V_1=(V_0^{-1}+X'X)^{-1},
\quad
m_1=V_1(V_0^{-1}m_0+X'Y).
\]
Stationarity is imposed by drawing from this posterior subject to the usual
root condition for the autoregressive polynomial.

\subsection{Sampling the latent states}

Here we prove that the sampler for the latent states described in the main paper is exact in the sense that it produces draws from the $p\!\left(S_t \mid S_{t+1:T},\,  y_{1:T},\, \theta\right)$ for any $t.$ It is sufficient to verify that the following relationship holds, up to a proportionality constant that does not depend on \(S_t\): \[
p\!\left(S_t \mid S_{t+1:T},\,  y_{1:T},\, \theta\right)
\;\propto\;
p\!\left(S_{t:t+k}\mid  y_{1:t+k},\, \theta\right).
\]
We begin by writing
\begin{align*}
p\!\left(S_t \mid S_{t+1:T},\,  y_{1:T},\, \theta\right)
=
\frac{
p\!\left(S_t,\, S_{t+1:T},\,  y_{1:T}\mid \theta\right)
}{
p\!\left(S_{t+1:T},\,  y_{1:T}\mid \theta\right)
} \propto
p\!\left(S_t,\, S_{t+1:T},\,  y_{1:T}\mid \theta\right),
\end{align*}
since the denominator does not depend on \(S_t\). Below, when \(k=0\), \(S_{t+1:t+k}\) and \(S_{t+k+1:T}\) are interpreted as \(S_{t+1}\) and \(S_{t+2:T}\), respectively. When \(k>0\), the notation is standard. We have
\begin{align*}
p\!\left(S_t,\, S_{t+1:T},\,  y_{1:T}\mid \theta\right)
&=
p\!\left(S_t,\, S_{t+1:t+k},\,  y_{1:t+k}\mid \theta\right) \\
&\quad\times
p\!\left(S_{t+k+1:T},\,  y_{t+k+1:T}
\mid S_t,\, S_{t+1:t+k},\,  y_{1:t+k},\, \theta\right) \\
&\propto
p\!\left(S_t,\, S_{t+1:t+k},\,  y_{1:t+k}\mid \theta\right),
\end{align*}
where the proportionality follows because the second factor does not depend on \(S_t\), since \(S_{t+1:t+k}\) already contains the most recent \(k\) state values. Finally,
\begin{align*}
p\!\left(S_t,\, S_{t+1:t+k},\,  y_{1:t+k}\mid \theta\right)
&=
p\!\left(S_t,\, S_{t+1:t+k}\mid  y_{1:t+k},\, \theta\right)
\, p\!\left( y_{1:t+k}\mid \theta\right) \\
&\propto
p\!\left(S_t,\, S_{t+1:t+k}\mid  y_{1:t+k},\, \theta\right).
\end{align*}
Therefore,
\[
p\!\left(S_t \mid S_{t+1:T},\,  y_{1:T},\, \theta\right)
\;\propto\;
p\!\left(S_t,\, S_{t+1:t+k}\mid  y_{1:t+k},\, \theta\right).
\]

\subsection{Sampling the transition probabilities}
We use independent Beta distribution priors for $(p_{11},p_{22})$:
\begin{equation*}
p_{11}\sim Beta(\alpha _{1},\beta _{1}) \text{ and }p_{22}\sim Beta(\alpha
_{2},\beta _{2}).
\end{equation*}%
Their joint density satisfies 
\begin{equation*}
f(p_{11},p_{22})\propto p_{11}^{\alpha _{1}-1}(1-p_{11})^{\beta
_{1}-1}p_{22}^{\alpha _{2}-1}(1-p_{22})^{\beta _{2}-1}.
\end{equation*}%
The likelihood function for $(p_{11},p_{22})$, given $S_{1:T}$, is 
\begin{equation*}
L(p_{11},p_{22}|S_{1:T})=p_{11}^{n_{11}}(1-p_{11})^{n_{12}}p_{22}^{n_{22}}(1-p_{22})^{n_{21}},
\end{equation*}%
where $n_{ij}$ refers to the number of transitions from state $i$ to state $%
j $ based on $S_{1:T}$. Thus, the posterior distribution of $%
(p_{11},p_{22})$ given $S_{1:T}$ is 
\begin{align*}
f(p_{11},p_{22}|S_{1:T})& \propto f(p_{11},p_{22})L(p_{11},p_{22}|%
S_{1:T}) \\
& =p_{11}^{\alpha _{1}+n_{11}-1}(1-p_{11})^{\beta
_{1}+n_{12}-1}p_{22}^{\alpha _{2}+n_{22}-1}(1-p_{22})^{\beta _{2}+n_{21}-1}.
\end{align*}%
Consequently, 
\begin{equation*}
p_{11}|S_{1:T}\sim Beta(\alpha _{1}+n_{11},\beta _{1}+n_{12}),\text{
and }p_{22}|S_{1:T}\sim Beta(\alpha _{2}+n_{22}, \beta _{2}+n_{21}).
\end{equation*}%
Therefore, given $S_{1:T}$, $p_{11}$ and $p_{22}$ can be
drawn from $Beta(\alpha _{1}+n_{11},\beta _{1}+n_{12})$ and $Beta(\alpha
_{2}+n_{22},\beta _{2}+n_{21})$ separately.

\newpage

\section{Data Sources and Component Structure}
\addcontentsline{toc}{section}{Appendix \thesection. Data Sources and Component Structure}

\paragraph{Sources.}
GDP component series are from FRED, Federal Reserve Bank of St.~Louis (Table 1.5.1, Percent Change From Preceding Period in Real Gross Domestic Product, Expanded Detail; \url{https://fred.stlouisfed.org/release/tables?rid=53&eid=13498}). Industrial production by industry and capacity utilization are from the Federal Reserve System (\url{https://www.federalreserve.gov/releases/g17/Current/default.htm}). Nonfarm payroll employment is from the Bureau of Labor Statistics (\url{https://www.bls.gov/data/}). Average weekly hours and overtime of production and nonsupervisory employees on private nonfarm payrolls are from FRED (Table B-7; \url{https://fred.stlouisfed.org/release/tables?rid=50&eid=5923}). NBER recession indicators are from \url{https://www.nber.org/cycles.htm}.

\paragraph{Frequency conversion.}
Monthly payroll and industrial production data are converted to quarterly frequency by simple averaging; other series are available quarterly.

\paragraph{Component structure.}
The structure of the GDP components is summarized in Figure~\ref{fig:gdp_dis}, and the components of the four business cycle indicators in Figures~\ref{fig:data_structure1} and~\ref{fig:data_structure2}, where the disaggregated variables used for estimation are shown in italics in blue.

\paragraph{Real GDP.}
Real GDP (excluding government) is decomposed into four main components: personal consumption expenditures, gross private domestic investment, exports, and imports. Personal consumption is split into goods, with eight disaggregate series covering durable and nondurable goods, and services, with eight series covering household consumption and NPISHs. Gross private domestic investment covers fixed investment, split into nonresidential (comprising structures, equipment with five series, and intellectual property products with three series) and residential. Exports and imports are each split into goods and services.

\paragraph{Industrial production.}
Industrial production is categorized by industry, with 21 components: durable goods manufacturing (11 series), nondurable goods manufacturing (8 series), mining, and utilities. The base year is 2012.

\paragraph{Capacity utilization.}
Capacity utilization, expressed as a percentage, is classified by industry. We consider 14 disaggregate series: wood product; nonmetallic mineral product; primary metal; fabricated metal product; machinery; computer and electronic product; electrical equipment, appliance, and component; transportation equipment; furniture and related product; miscellaneous; nondurable manufacturing; other manufacturing; mining; and electric and gas utilities.

\paragraph{Nonfarm payroll employment.}
We consider the classification by sector, with 15 disaggregate series: mining and logging; construction; durable goods; nondurable goods; wholesale trade; retail trade; transportation; utilities; information; financial activities; professional and business services; education and health services; leisure and hospitality; other services; and government. The unit is thousands of employees.

\paragraph{Average weekly hours.}
Hours and overtime of production and nonsupervisory employees are also classified by sector. We consider 14 disaggregate series: mining and logging; construction; durable manufacturing; nondurable manufacturing; trade (wholesale and retail); information; financial activities; professional and business services; education and health services; leisure and hospitality; transportation; utilities; and other services.

\section{Figures and Diagrams}
This section contains a figure for smoothed US recession probabilities from FRED and diagrams that further illustrate the structure of the disaggregated dataset.
\begin{figure}[H]
\caption{Components of Real GDP}
\vspace{3mm}
\includegraphics[width=1.9\linewidth]{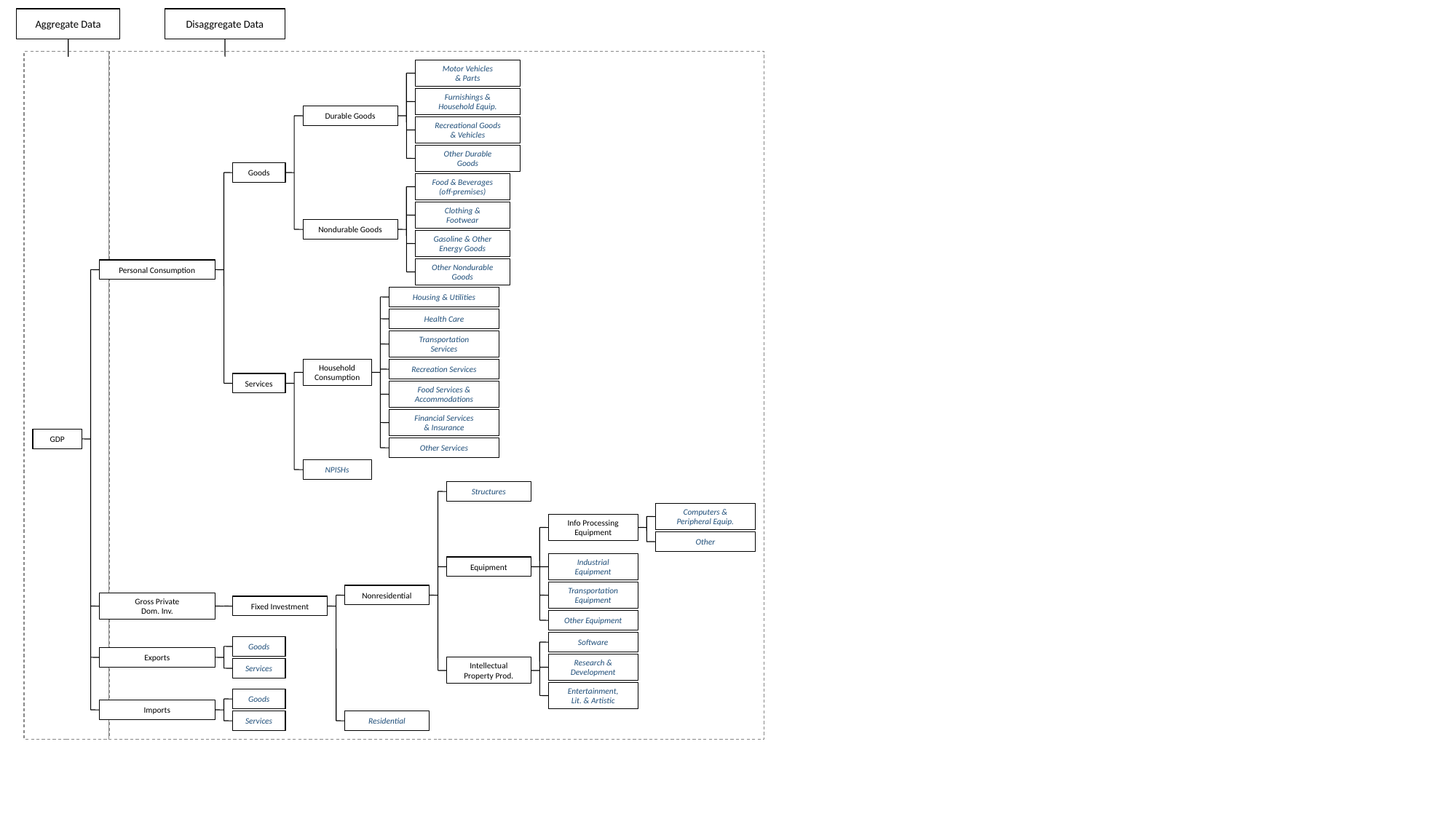}
\vspace{-2mm}
\begin{minipage}{1.0\linewidth}
\small\textit{Note:} The figure displays the GDP components excluding government spending. Components highlighted in light blue are those included in the estimation.
\end{minipage}
\label{fig:gdp_dis}
\end{figure}
\newpage
\begin{figure}[H]
\caption{Components of Four Economic Indicators: IP}
\vspace{1cm}
\hspace{2cm}
\includegraphics[width=1.6\linewidth, trim=0cm 3cm 10cm 0cm, clip]{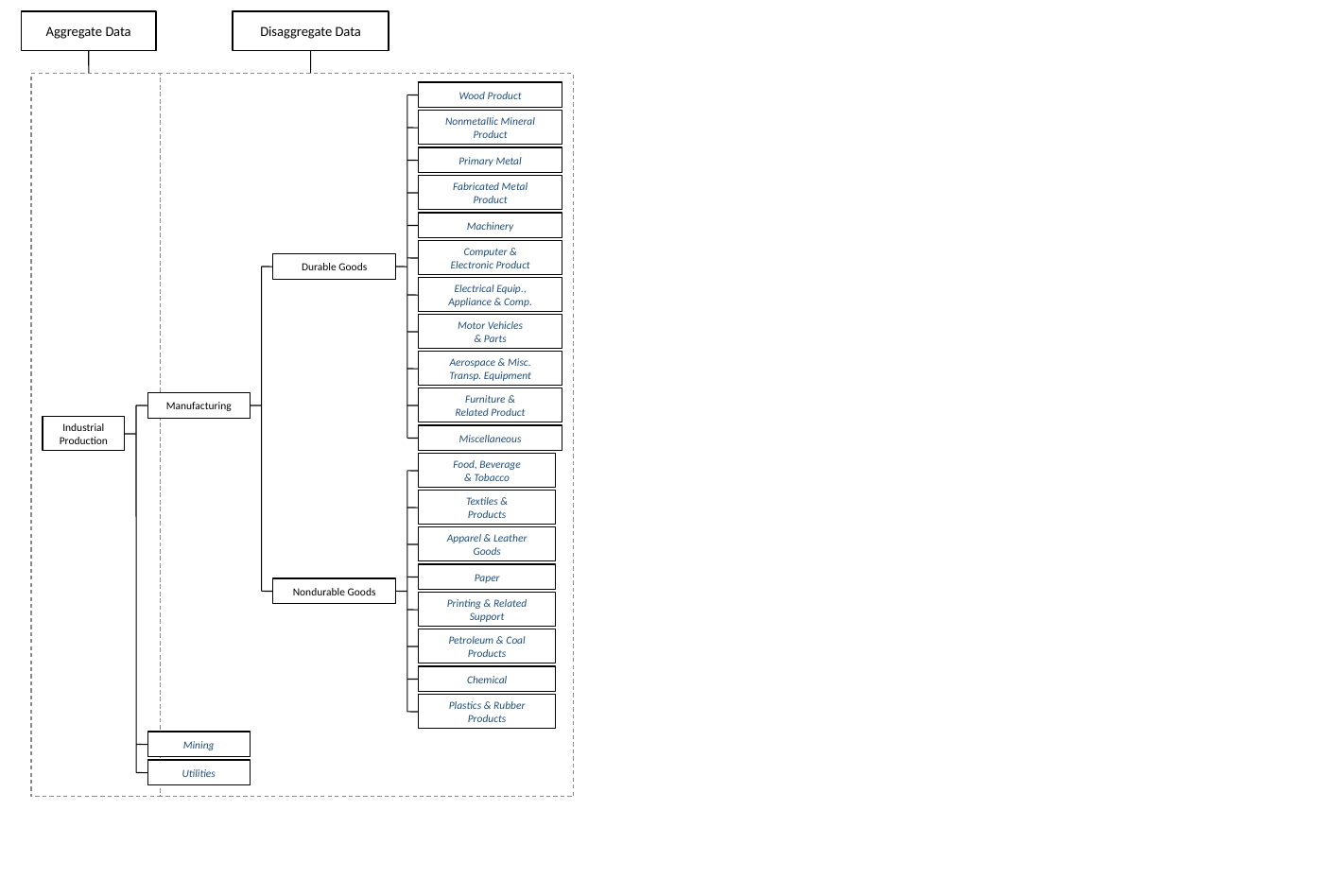}
\label{fig:data_structure1}
\begin{minipage}{1.0\linewidth}
\singlespacing
\small\textit{Note:} The figure displays the components of industrial production, one of four sets of economic indicators used in the analysis. Components highlighted in light blue are those included in the estimation.
\end{minipage}
\end{figure}

\newpage
\begin{figure}[H]
\caption{Components of Four Economic Indicators (cont'd)}
\hspace{1cm}
\includegraphics[width=1.6\linewidth, trim=0cm 0cm 9cm 0cm, clip]{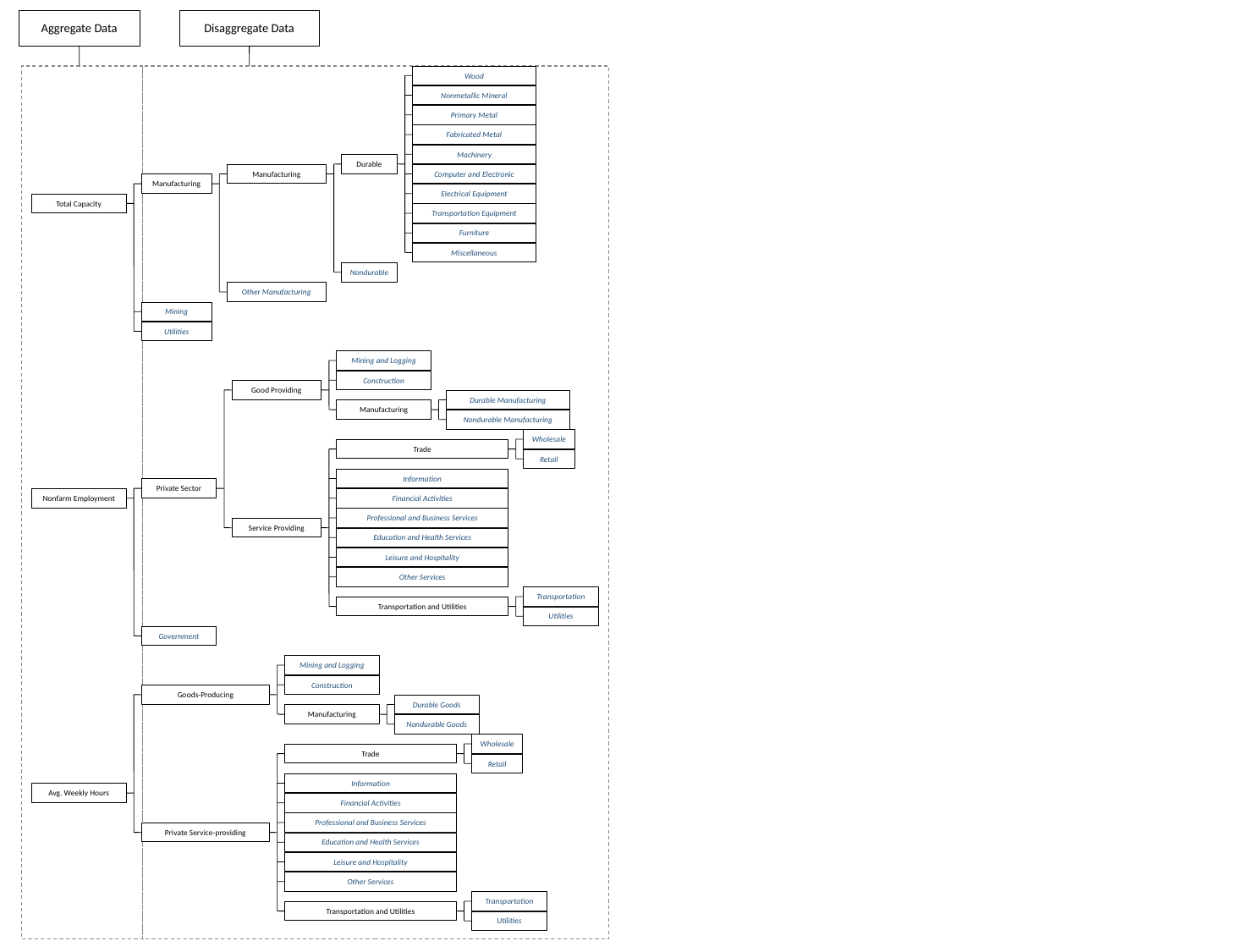}
\label{fig:data_structure2}
\begin{minipage}{1.0\linewidth}
\singlespacing
\small\textit{Note:} The figure displays the components of three economic indicators used in the analysis. Components highlighted in light blue are those included in the estimation.
\end{minipage}
\end{figure}


\begin{figure}[H]
    \centering

    \caption{Smoothed U.S. Recession Probability (FRED)}
    \label{fig:Fred_recession}

    \includegraphics[width=0.95\textwidth, clip, trim=0cm 9cm 0cm 9.5cm]{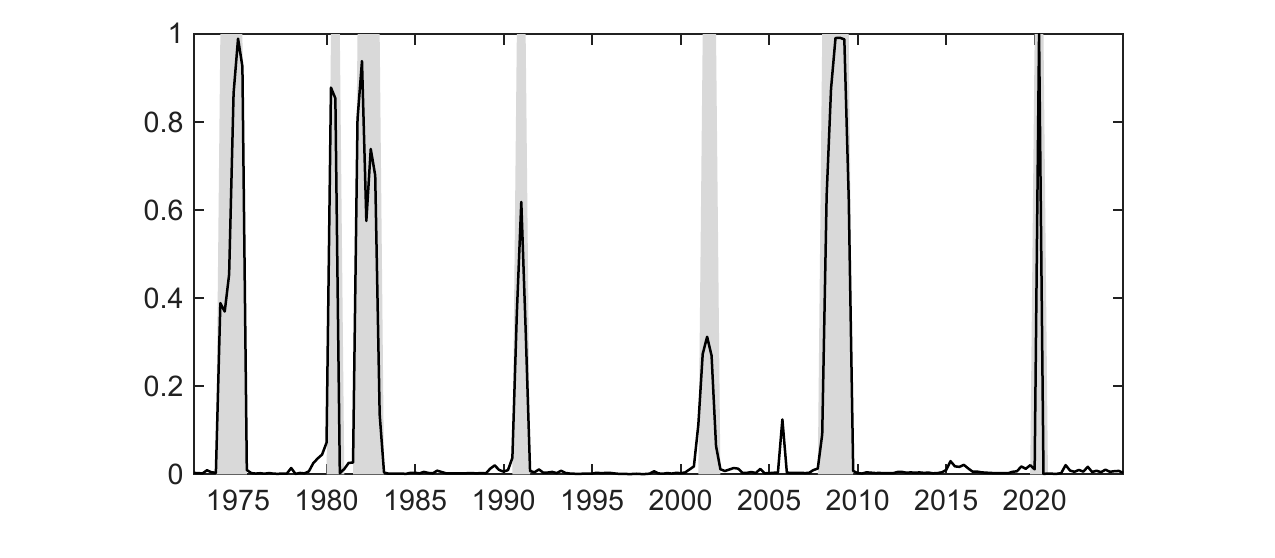}

    \vspace{0.3em}
  
\vspace{-2mm}
\begin{minipage}{0.9\linewidth}
\singlespacing
\small\textit{Source:} Chauvet, Marcelle and Piger, Jeremy Max,
    Smoothed U.S. Recession Probabilities [RECPROUSM156N], retrieved from FRED,
    Federal Reserve Bank of St. Louis;
    \url{https://fred.stlouisfed.org/series/RECPROUSM156N}. Smoothed recession probabilities are obtained from a dynamic-factor markov-switching model applied to four monthly coincident variables: non-farm payroll employment, the index of industrial production, real personal income excluding transfer payments, and real manufacturing and trade sales. This model was originally developed in Chauvet (1998).\end{minipage}
\end{figure}


\end{document}